\documentclass[12pt]{article}

\usepackage{latexsym}
\usepackage{amsmath}
\usepackage{amssymb}
\usepackage{amsthm}
\usepackage{amscd}
\usepackage[mathscr]{eucal}
\usepackage{fullpage}
\usepackage{graphicx}
\usepackage{subfigure}
\usepackage{psfrag}
\usepackage{rotating}
\usepackage{lscape} 
\usepackage{color}

\usepackage{adjustbox}

\newcommand{\textfrac}[2]{{\textstyle \frac{#1}{#2}}}

\newcommand{\R}{\mathbb{R}}

\newcommand{\prad}{\widehat{p}_\mathrm{rad}}
\newcommand{\ptan}{\widehat{p}_\mathrm{tan}}
\newcommand{\Prad}{\mathcal{P}_\mathrm{rad}}
\newcommand{\Ptan}{\mathcal{P}_\mathrm{tan}}

\definecolor{ggreen}{cmyk}{0.7,     0,      0.9,      0}
\definecolor{viol}{cmyk}{0.3,1,0,0}
\definecolor{myred}{cmyk}{0.1, 1, 0.5, 0}
\definecolor{bblue}{rgb}{0.2, 0.29996, 0.8 }

\theoremstyle{plain}

\newtheorem{theorem}{Theorem}
\newtheorem{lemma}{Lemma}

\newtheorem{proposition}{Proposition}

\newtheorem{definition}{Definition}
\newtheorem{assumption}{Assumption}

\title{\bf 
	 Static self-gravitating \\ Newtonian elastic balls}
 \author
 {A.~Alho  \\
           {\small Center for Mathematical Analysis, Geometry and Dynamical Systems}  \\
      {\small Instituto Superior T\'ecnico, Universidade de Lisboa} \\
      {\small Av. Rovisco Pais, 1049-001 Lisboa, Portugal} \\
       {\small  \tt   artur.alho@tecnico.ulisboa.pt}\\[0.4cm]
       S. Calogero  \\
       {\small Department of Mathematical Sciences}  \\
       {\small Chalmers University of Technology, University of Gothenburg}  \\
       {\small Gothenburg, Sweden} \\
       {\small  \tt  calogero@chalmers.se}
       }

\begin{document}
\maketitle


\begin{abstract}
\noindent
The existence of static self-gravitating Newtonian  elastic balls is proved under general assumptions on the constitutive equations of the elastic material. The proof uses methods from the theory of finite-dimensional dynamical systems and the Euler formulation of elasticity theory for spherically symmetric  bodies introduced recently by the authors. Examples of elastic materials covered by the results of this paper are Saint Venant-Kirchhoff, John and Hadamard materials.  

\end{abstract}

\section{Introduction and main results}
Static, spherically symmetric, self-gravitating matter distributions in Newtonian gravity are described by the following system of nonlinear integro-differential equations
\begin{subequations}\label{CPSS}
	\begin{align}
	\frac{dp_{\mathrm{rad}}}{dr} &= -\frac{2}{r}(p_{\mathrm{rad}}-p_{\mathrm{tan}})-\rho \frac{d\Phi}{dr}, \\
	\frac{d\Phi}{dr} &= \frac{m}{r^2},\quad m(r)=4\pi\int^{r}_{0}\rho(s)s^2 ds 
	\end{align}
\end{subequations}
defined in the interior of the matter support, $\Omega:=\mathrm{Int}\{r>0:\rho(r)>0\}$. Here, $p_{\mathrm{rad}}(r)$ and  $p_{\mathrm{tan}}(r)$ are the radial and tangential stresses (or pressures) respectively, $\rho(r)$ is the mass density,  $\Phi(r)$ is the Newtonian gravitational potential self-induced by the matter distribution and $m(r)$ is the mass enclosed in the ball of radius $r$.

The system~\eqref{CPSS} is completed by adding a constitutive equation between the pressure variables and the mass density which depends on the specific material model and its response to stress. A popular example is the barotropic fluid, for which $p_{\mathrm{rad}}(r)=p_{\mathrm{tan}}(r)=\widehat{p}(\rho(r))$, for some function $\widehat{p}:(0,\infty)\to\R$. For this matter model, the existence of solutions to~\eqref{CPSS} with finite radius has been extensively studied and is by now well-understood~\cite{HU03,Makino84, RR}. A non-linear stability result for these solutions is proved in~\cite{Rein}. 
Barotropic fluids are widely used in astrophysics, e.g., to model the matter content of main sequence stars~\cite{Cha39,KWW}.

In this paper we study the system~\eqref{CPSS} when the matter distribution consists of a single elastic body, which we assume to have the shape of a ball. 
Examples of astrophysical objects which are frequently modeled as elastic bodies are planets and neutron star crusts~\cite{CH2008,MW16}.
For our analysis we use the Euler form of the the constitutive equations for spherically symmetric elastic balls given recently in~\cite{AC18}, namely

\begin{subequations}\label{elasticEOS}
\begin{equation}
p_{\mathrm{rad}}(r)=\prad(\delta(r),\eta(r)),\quad p_{\mathrm{tan}}(r)=\ptan(\delta(r),\eta(r)),
 \end{equation}
 where
 \begin{equation}\label{etadeltadef}
 \delta(r)=\frac{\rho(r)}{\mathcal{K}},\quad\eta(r)=\frac{m(r)}{\frac{4\pi}{3}\mathcal{K} r^3},
 \end{equation}
 \end{subequations}
 where $\mathcal{K}>0$ is a given constant (reference density) and $\prad,\ptan:(0,\infty)^2\to\R$ are functions independent of $\mathcal{K}$ (constitutive functions). The reference configuration of the elastic ball 
 corresponds to the state $\delta=\eta=1$ in which the mass density of the ball equals the constant reference density $\mathcal{K}$. 
In the case of so-called hyperelastic materials the constitutive functions can be deduced from a stored energy function $\widehat{w}(\delta,\eta)$ by
\begin{equation}\label{hyperDef}
\prad(\delta,\eta) = \delta^2 \partial_\delta \widehat{w}(\delta,\eta), \quad  \quad \ptan(\delta,\eta)=\prad(\delta,\eta)+\frac{3}{2}\delta\eta \partial_{\eta} \widehat{w}(\delta,\eta).
\end{equation}
When $\widehat{w}$ is independent of $\eta$, the hyperelastic material becomes a barotropic fluid. 
The system~\eqref{CPSS} in terms of the variables $(\delta(r),\eta(r))$ reads
 	\begin{subequations}\label{CPSS_Rho}
 		\begin{align}
 	\partial_\delta \prad(\delta,\eta)	\frac{d\delta}{dr} &=-\frac{3}{r}\partial_{\eta}\prad(\delta,\eta)(\delta-\eta) -\frac{2}{r}(\prad(\delta,\eta)-\ptan(\delta,\eta))-\frac{4\pi}{3}\mathcal{K}^2r \delta \eta,\label{eqdelta}\\
 	\frac{d\eta}{dr} &= \frac{3}{r}(\delta-\eta) .\label{eqeta}
 		\end{align}
 		\end{subequations}
In \cite{AC18}, and using this form of the equations, the particular example of the (non-hyperelastic) Seth model was studied in detail. The existence of single and multi-body configurations consisting of a ball, or a vacuum core shell, surrounded by an arbitrary number of shells was proved along with sharp mass/radius inequalities. In this paper we extend the results in~\cite{AC18} for single Seth elastic balls to general constitutive equations. Before stating our main results, we give the precise definition of regular and strongly regular ball of matter.

\begin{definition}\label{balldef}
	A triple $(\rho,p_\mathrm{rad},p_\mathrm{tan})$ is called a regular static self-gravitating ball of matter if there exists a constant $R>0$ such that $\Omega:=\mathrm{Int}\{r>0:\rho(r)>0\}=(0,R)$ and
	\begin{itemize}
		\item[(i)] $(\rho,p_\mathrm{rad},p_\mathrm{tan})\in C^0([0,R])\cap C^1((0,R])$ satisfy~\eqref{CPSS} for $r\in (0,R)$,
		\item[(ii)] $p_\mathrm{rad}(r),p_\mathrm{tan}(r)$ are positive for $r\in [0,R)$, 
		\item[(iii)] $p_\mathrm{rad}(R)=0$,
		\item[(iv)] $p_\mathrm{rad}(0)=p_\mathrm{tan}(0)$,
		\item[(v)] $\rho(r)=p_\mathrm{tan}(r)=p_\mathrm{tan}(r)=0$, for $r>R$.
	\end{itemize}
If in addition $(\rho,p_\mathrm{rad},p_\mathrm{tan})\in C^1([0,R])$ and 
\[
\lim_{r\to 0^+}\frac{d\rho}{dr}(r)=\lim_{r\to 0^+}\frac{dp_\mathrm{rad}}{dr}(r)=\lim_{r\to 0^+}\frac{dp_\mathrm{tan}}{dr}(r)=0,
\]
then $(\rho,p_\mathrm{rad},p_\mathrm{tan})$ is called a strongly regular static self-gravitating ball of matter. When the principal pressures  are given as in~\eqref{elasticEOS} for some constant $\mathcal{K}>0$ and functions $\prad,\ptan\in C^1((0,\infty)^2)$ independent of $\mathcal{K}$, then we call $(\rho,p_\mathrm{rad},p_\mathrm{tan})$ a (strongly) regular static self-gravitating elastic ball with reference density $\mathcal{K}$ and constitutive functions $\prad,\ptan$.
\end{definition}
%
 %


{\it Remark.} The boundary condition $p_\mathrm{rad}(R)=0$ means that the ball of matter is surrounded by vacuum. See~\cite{KWW} for other boundary conditions used in astrophysics. We also remark that $\lim_{r\to 0^+}\delta(r)=\lim_{r\to 0^+}\eta(r)$ holds for regular elastic balls.

The distinction between regular and strongly regular balls is important for the following reason. 
Define the mass density $\widetilde{\rho}(x)$ and the stress tensor $\sigma(x)$ in Cartesian coordinates by $\widetilde{\rho}(x)=\rho(|x|)$ and
\[
\sigma_{ij}(x)=-p_\mathrm{rad}(|x|)\frac{x_ix_j}{|x|^2}-p_\mathrm{tan}(|x|)\left(\delta_{ij}-\frac{x_ix_j}{|x|^2}\right).
\]
Then $\widetilde{\rho},\sigma\in C^0(B_R)\cap C^1(B_R\diagdown\{0\})$ for regular balls supported in $B_R=\{x:|x|\leq R\}$, while $\widetilde{\rho},\sigma\in C^1(B_R)$ for strongly regular balls. For barotropic fluids with equation of state $p_\mathrm{rad}=p_\mathrm{tan}=\widehat{p}(\rho)$ it follows immediately from~\eqref{CPSS} that any regular static self-gravitating ball is strongly regular if $\widehat{p}\,'(\rho_c)\neq 0$, where $\rho_c=\rho(0)$ is the central density of the ball, which holds in particular under the strict hyperbolicity condition $\widehat{p}\,'(\rho)>0$ for the Euler equations. Our first main result shows that a 
similar strong regularity criteria holds for regular elastic balls. We set
\[
\chi(\delta)=\ptan(\delta,\delta)-\prad(\delta,\delta),
\]
and denote
\begin{equation}\label{isopressdef}
\widehat{p}_\mathrm{iso}(\delta,\eta)=\frac{1}{3}\big(\widehat{p}_\mathrm{rad}(\delta,\eta)+2\ptan(\delta,\eta)\big)=-\frac{1}{3}\mathrm{Tr}(\sigma);
\end{equation}
the quantity $p_\mathrm{iso}(r)=\widehat{p}_\mathrm{iso}(\delta(r),\eta(r))$ is the isotropic pressure of the ball.


%
%
\begin{theorem}\label{regtheo}
Let $(\rho,p_\mathrm{rad},p_\mathrm{tan})$ be a regular static self-gravitating elastic ball with central density $\rho(0)=\rho_c$. Assume that $\prad,\ptan\in C^2((0,\infty)^2)$ and 
\begin{equation}\label{chi}
\chi(\delta)=0,\ \partial_\eta \widehat{p}_\mathrm{iso}(\delta,\delta)=0,\quad \text{for all $\delta>0$.}
\end{equation}
If $\partial_\delta\prad(\delta_c,\delta_c)\neq 0$, where $\delta_c=\rho_c/\mathcal{K}$,
then $(\rho,p_\mathrm{rad},p_\mathrm{tan})$ is strongly regular and the following estimate holds:
\begin{equation}\label{decaypressurescenter}
|p'_\mathrm{rad}(r)|+|p'_\mathrm{tan}(r)|+|\delta'(r)|+|\eta'(r)|\leq C r, \quad 0\leq r\leq 1,
\end{equation}
for some positive constant $C$.  Moreover in the case of hyperelastic materials there holds the identity
\begin{equation}\label{isopressid}
\frac{3}{2}\partial_\eta \widehat{p}_\mathrm{iso}(\delta,\delta)=\chi'(\delta)-\frac{\chi(\delta)}{\delta},
\end{equation}
hence for hyperelastic materials the condition~\eqref{chi} is equivalent to $\chi(\delta)=0$, for all $\delta>0$.
\end{theorem}
The proof of Theorem~\ref{regtheo} is given in Section~\ref{proofregtheo} where we also show that the assumptions are satisfied by the Seth model studied in~\cite{AC18}. In particular the regular elastic balls constructed in~\cite{AC18} are strongly regular. We remark that the condition $\chi(\delta)=0$ in Theorem~\ref{regtheo} is equivalent to demand that the principal pressures should be equal at the center for all possible values of the central density. This condition is satisfied by all physically relevant elastic materials, including of course barotropic fluids. Furthermore we show in~\cite{AC19} that the condition for strict hyperbolicity of the time dependent equations of motion for spherically symmetric elastic bodies is that
\begin{equation}\label{stricthypcond}
\partial_\delta\prad(\delta(r),\eta(r))>0
\end{equation}
should hold for all $r\geq0$, which implies in particular the assumption $\partial_\delta\prad(\delta_c,\delta_c)\neq 0$ in Theorem~\ref{regtheo}.

The second main result of this paper is the following.
\begin{theorem}\label{maintheo}
Under the Assumptions~\ref{As1}--\ref{neg-press} on the constitutive functions $\prad,\ptan$ given in Section~\ref{assumptions}, there exists $\Delta\in (1,\infty]$, uniquely determined by $(\prad,\ptan)$ and which can be explicitly computed, such that for all $\rho_c,\mathcal{K}>0$ satisfying
\begin{equation}\label{initialdata}
1<\delta_c:=\frac{\rho_c}{\mathcal{K}}<\Delta
\end{equation} 
there holds 
\begin{equation}\label{central-conditions}
\prad(\delta_c,\delta_c)>0,
\quad \partial_\delta \prad(\delta_c,\delta_c)>0
\end{equation}
and there exists a unique regular static self-gravitating elastic ball with central density $\rho(0)=\rho_c$ and reference density $\mathcal{K}$. Moreover 
\begin{equation}\label{bounds}
\partial_\delta\prad(\delta(r),\eta(r))>0,\quad \rho(r)<\rho_c,\quad \frac{4\pi}{3}\max(\rho(r),\mathcal{K})r^3<m(r)<\frac{4\pi}{3}\rho_c r^3
\end{equation}
hold for all $r\in (0,R]$, where $R>0$ is the radius of the ball. 
\end{theorem}

The proof of Theorem~\ref{maintheo} is given in Section~\ref{proof}; it is independent of the proof of Theorem~\ref{regtheo}. In Section~\ref{Examples} we show that the Assumptions~\ref{As1}--\ref{neg-press} on $\prad,\ptan$ presented in Section~\ref{assumptions} are satisfied by important and widely used examples of elastic material models, namely Saint-Venant-Kirchhoff materials, John materials and Hadamard materials. We also give for each specific material model a more precise formulation of Theorem~\ref{maintheo}, which contains the exact value of $\Delta$. In particular we shall find that $\Delta=\infty$ for the John model, while $\Delta<\infty$ for the Saint-Venant-Kirchhoff and Hadamard models. Moreover in a remark after Theorem~\ref{hadtheo} in Section~\ref{Examples} we argue that the constant $\Delta$ might not be optimal for Hadamard materials.

{\it Remark.}  Most of the elastic material models used in the applications, including the examples in Section~\ref{Examples}, belong to the class of power-law hyperelastic materials, see Definition~\ref{powerdef} in Section~\ref{powerlawmaterials}. For these materials some of the assumptions in Theorem~\ref{maintheo} are always satisfied, see Proposition~\ref{powerprop} in Section~\ref{powerlawmaterials}. Moreover all hyperelastic power-law materials satisfy the hypotheses of Theorem~\ref{regtheo} and thus when Theorem~\ref{maintheo} is applied to these models, the regular static self-gravitating elastic ball is strongly regular.

{\it Remark.} The constant $\Delta$ in~\eqref{initialdata} is determined uniquely by the constitutive equations as follows. By Assumption~\ref{initial-reg} there exists $\Delta_\flat\in (1,\infty]$ such that $\partial_\delta \prad(\delta,\delta)>0$ for $0<\delta<\Delta_\flat$ and if $\Delta_\flat<\infty$ then $\partial_\delta\prad(\Delta_\flat,\Delta_\flat)=0$. By Assumption~\ref{ptanpositive} there exists $\Delta_\sharp\in(1,\infty]$ such that $0<\eta<\Delta_\sharp$, $0<\delta<\eta$ and $\prad(\delta,\eta) \geq 0$ imply $\ptan(\delta,\eta)>0$ and if $\Delta_\sharp<\infty$ then $\prad(\delta_*,\Delta_\sharp)=0$ for some $\delta_*\in (0,\Delta_\sharp)$ implies $\ptan(\delta_*,\Delta_\sharp)=0$; note that $\Delta_\flat,\Delta_\sharp$ are both unique. The value of $\Delta$ is given by $\Delta=\min(\Delta_\flat,\Delta_\sharp)$. The bound $\delta_c<\Delta_\flat$ ensures that the strict hyperbolicity condition~\eqref{stricthypcond} is verified at the center, while the bound $\delta_c<\Delta_\sharp$ ensures that the tangential pressure remains positive in the interior of the ball.
For the examples of stored energy functions in Section~\ref{Examples} the values of $\Delta_\flat,\Delta_\sharp$ can be computed exactly; for more complicated material models these values may be only found numerically. 

{\it Remark.} The first inequality in~\eqref{central-conditions} is equivalent to the positivity of the central radial pressure of the ball, and it is therefore necessary. As shown in Section~\ref{Examples}, for Saint Venant-Kirchhoff and Hadamard materials there holds $\widehat{p}_\mathrm{rad}(\delta_c,\delta_c)\leq 0$ for $\delta_c\leq 1$, hence the assumption $\delta_c>1$ in Theorem~\ref{maintheo} is necessary for these materials. However for the John model there exists $\Delta_*<1$ such that $\widehat{p}_\mathrm{rad}(\delta_c,\delta_c)> 0$ for $\delta_c\in (0,\Delta_*)$, see the remark after Theorem~\ref{johntheo} in Section~\ref{Examples}.  The question whether finite radius ball solutions for the John model exist when $\delta_c\in (0,\Delta_*)$ remains open. 

The proof of Theorem~\ref{maintheo} is based on the following argument. We begin by proving that, under Assumptions~\ref{As1}--\ref{hardone} on the constitutive functions $\prad,\ptan$, center data $\delta(0)=\eta(0)=\delta_c$ satisfying $\partial_\delta\prad(\delta_c,\delta_c)>0$  launch a unique global positive solution $(\delta,\eta)$ of~\eqref{CPSS_Rho} and $\delta(r)\to 0$ as $r\to \infty$, see Theorem~\ref{global-regular} in Section~\ref{proof}. The proof of this result relies on methods from the theory of finite-dimensional dynamical systems~\cite{H02}. The regular balls in Theorem~\ref{maintheo} are constructed by truncating these global solutions as follows. By Assumption~\ref{neg-press} on $\prad$ the radial pressure becomes negative before the density approaches zero.
Therefore if $p_\mathrm{rad}(0)=\prad(\delta_c,\delta_c)>0$, then there exists $R>0$ such that $\overline{p}_\mathrm{rad}(r)=\prad(\delta(r),\eta(r))>0$ for $r\in [0,R)$, while $\overline{p}_\mathrm{rad}(R)=0$. By Assumption~\ref{ptanpositive} on $\prad,\ptan$ the tangential pressure is positive when the radial pressure is positive, hence we also have $\overline{p}_\mathrm{tan}(r)=\ptan(\delta(r),\eta(r))>0$, for $r\in [0,R)$. Thus letting $\overline{\rho}(r)=\mathcal{K}\delta(r)$, the triple $(\rho,p_\mathrm{rad},p_\mathrm{tan})=(\overline{\rho},\overline{p}_\mathrm{rad},\overline{p}_\mathrm{tan})\mathbb{I}_{r\leq R}$ defines a regular static self-gravitating elastic ball supported in the interval $[0,R]$. 
 
%

\section{Proof of Theorem~\ref{regtheo}}\label{proofregtheo}
As we are interested only in the behavior of regular balls when $r\to0^+$, we may assume that $r\in [0,r_*)$, where $r_*$ can be chosen arbitrarily small. In particular, we can assume that $(\delta,\eta)$ lies in an arbitrarily small disk $D$ around $(\delta_c,\delta_c)$.  Hence by Taylor's theorem there exist functions $\mathcal{R}_1(\delta,\eta)$, $\mathcal{R}_{2}(\delta,\eta)$ bounded in $D$ such that
\begin{align*}
&\partial_\eta\prad(\delta,\eta)=\partial_\eta\prad(\delta,\delta)+\mathcal{R}_1(\delta,\eta)(\delta-\eta),\\
&\prad(\delta,\eta)-\ptan(\delta,\eta)=-(\partial_\eta\prad(\delta,\delta)-\partial_\eta\ptan(\delta,\delta))(\delta-\eta)+\mathcal{R}_2(\delta,\eta)(\delta-\eta)^2,
\end{align*}
where we used that $\chi(\delta)=\ptan(\delta,\delta)-\prad(\delta,\delta)=0$.  Replacing in~\eqref{eqdelta} and using the hypothesis $3\partial_\eta \widehat{p}_\mathrm{iso}(\delta,\delta)=\partial_\eta(\prad+2\ptan)(\delta,\delta)=0$, the first order terms cancel out and thus we obtain
\[
\partial_\delta\prad(\delta,\eta)\frac{d\delta}{dr}=(\mathcal{R}_1(\delta,\eta)+\mathcal{R}_2(\delta,\eta))\frac{(\delta-\eta)^2}{r}-\frac{4\pi\mathcal{K}^2}{3}\delta\eta r.
\]
Since $\partial_\delta\prad(\delta_c,\delta_c)\neq 0$, then for $r_*$ sufficiently small there holds $\inf_{r\in (0,r_*)}\partial_\delta\prad(\delta,\eta)\neq 0$. In conclusion, in a sufficiently small interval $(0,r_*)$ we can write the equation for $\delta'$ as
\begin{equation}\label{eqdeltanew}
\frac{d\delta}{dr}=\mathcal{R}(\delta,\eta)\left(\frac{(\delta-\eta)^2}{r}+r\right),
\end{equation}
where $\mathcal{R}(\delta,\eta)$ is bounded in $D$. Let $u(r)=r\delta'(r)\in C^0((0,r_*))$. By~\eqref{eqdeltanew}, $u$ extends uniquely to a continuous function on the interval $[0,r_*)$, which we continue to denote by $u$. Furthermore $u(r)\to 0$ as $r\to 0^+$.  Now, by the definition of $\eta$ in~\eqref{etadeltadef} we have
\[
\eta(r)=\frac{3\int_0^r\delta(s)s^2\,ds}{r^3}=\delta(r)-r^{-3}\int_0^ru(s)s^2\,ds.
\]
Using Cauchy-Schwarz's inequality we have
\begin{equation}\label{frometa}
|\delta(r)-\eta(r)|\leq \frac{1}{\sqrt{5r}}\left(\int_0^ru(s)^2\,ds\right)^{1/2}.
\end{equation}
Replacing in~\eqref{eqdeltanew} we obtain that $u(r)$ satisfies, for all $r\in (0,r_*)$, 
\[
u(r)\leq C\left(r^2+\frac{1}{r}\int_0^ru(s)^2\,ds\right), 
\]
where $C$ here and below denotes a positive constant which may change from line to line. Hence for all $\varepsilon\in (0,r_*)$ and $r\in [\varepsilon,r_*)$ we have
\begin{align*}
u(r)&\leq C\left(r^2+\frac{1}{r}\int_0^\varepsilon u(s)^2\,ds+\frac{1}{r}\int_\varepsilon^ru(s)^2\,ds\right)\\
&\leq C\left(r^2+\frac{\varepsilon}{r}D(\varepsilon)+\frac{1}{r}\int_\varepsilon^ru(s)^2\,ds\right),
\end{align*}
where $D(\varepsilon)=\|u\|_{L^\infty((0,\varepsilon))}^2\to 0$ as $\varepsilon\to 0$ and where the constant $C>0$ here and below is independent of $\varepsilon$.
We next use the following nonlinear Gr\"{o}nwall's type inequality which was proved in~\cite{stat}, see also~\cite[Theorem 25]{ssd} and~\cite[Corollary 2]{MS}.
\begin{lemma}[B.~Stachurska~\cite{stat}]
Let the functions $u,a,b$ and $k$ be continuous and nonnegative in $J = [\alpha,\beta]$, $n\geq 2$ be a positive integer and assume that $a/b$ is a nondecreasing function. If
\[
u(r)\leq a(r)+b(r)\int_\alpha^r k(s) u(s)^n\,ds,\quad r\in J,
\]
then
\[
u(r)\leq a(r)\left(1-(n-1)\int_\alpha^rk(s)b(s)a(s)^{n-1}\,ds\right)^{\frac{1}{1-n}},\quad \alpha\leq r\leq\beta_n,
\]
where
\[
\beta_n=\sup\left\{r\in J:1-(n-1)\int_\alpha^rk(s)b(s)a(s)^{(n-1)}\,ds>0\right\}.
\]
\end{lemma}
The result applies to our case with
\[
n=2,\quad a(r)=C(r^2+\frac{\varepsilon}{r}D(\varepsilon)),\quad b(r)=\frac{C}{r},\quad k(r)\equiv 1, \quad \alpha=\varepsilon,\quad \beta=r_*
\]
and gives
\[
u(r)\leq \frac{C(r^2+\varepsilon D(\varepsilon)/r)}{1-C^2(r^2-\varepsilon^2)/2+C^2(\varepsilon D(\varepsilon)/r-D(\varepsilon))}.
\]
This estimate holds in the largest interval $[\varepsilon, \beta_2(\varepsilon))\subseteq[\varepsilon,r_*)$ where the denominator is positive. As the latter quantity is bounded below by $1-C^2r^2-C^2D(\varepsilon)$, we obtain that $u$ satisfies 
\[
u(r)\leq \frac{C(r^2+D(\varepsilon))}{1-C^2r^2-C^2D(\varepsilon)},
\]
for $\varepsilon$ sufficiently small and for all $r\in [\varepsilon,r_*)$ such that $\varepsilon\leq r<C^{-1}\sqrt{1-C^2D(\varepsilon)}$. Taking the limit as $\varepsilon\to 0^+$ we conclude that $u(r)\leq Cr^2$ holds for all sufficiently small $r$, that is $\delta'(r)\leq C r$, as claimed in the theorem. Replacing in~\eqref{frometa} we obtain $|\delta-\eta|\leq C r^2$, and thus~\eqref{eqeta} gives the estimate $|\eta'(r)|\leq C r$. Finally, since
\[
p_\mathrm{rad}'(r)=\partial_\delta\prad(\delta,\eta)\delta'+\partial_\eta \prad(\delta,\eta)\eta',
\]
and similarly for the tangential pressure, the proof of~\eqref{decaypressurescenter} follows. As to the identity~\eqref{isopressid}, we compute, using~\eqref{hyperDef},
\[
\chi(\delta)=\frac{3}{2}\delta^2\partial_\eta\widehat{w}(\delta,\delta)\Rightarrow \chi'(\delta)=3\delta\partial_\eta \widehat{w}(\delta,\delta)+\textfrac{3}{2}\delta^2(\partial_\eta\partial_\delta\widehat{w}(\delta,\delta)+\partial_\eta^2\widehat{w}(\delta,\delta)),
\]
as well as $\widehat{p}_\mathrm{iso}(\delta,\eta)=\delta^2\partial_\delta\widehat{w}(\delta,\eta)+\delta\eta\partial_\eta\widehat{w}(\delta,\eta)$, which implies
\[
\partial_\eta \widehat{p}_\mathrm{iso}(\delta,\delta)=\delta^2\partial_\delta\partial_\eta\widehat{w}(\delta,\delta)+\delta\partial_\eta\widehat{w}(\delta,\delta)+\delta^2\partial_\eta^2\widehat{w}(\delta,\delta)=\textfrac{2}{3}(\chi'(\delta)-\chi(\delta)/\delta).\qed
\]

{\bf Example: The Seth model.} The constitutive equations for Seth materials are 
\begin{equation}\label{pseth}
\widehat{p}_\mathrm{rad}(\delta,\eta)=\lambda\,\eta^{2/3}+\frac{\lambda+2\mu}{2}\eta^{-4/3}\delta^2-p_0,\quad
\widehat{p}_\mathrm{tan}(\delta,\eta)=(\lambda+\mu)\,\eta^{2/3}+\frac{\lambda}{2}\eta^{-4/3}\delta^2-p_0,
\end{equation}
where $\lambda,\mu$ are the Lam\'e material constants, with $\mu>0$ and $p_0=(3\lambda+2\mu)/2>0$. For this material model the identities~\eqref{chi} hold and moreover 
\[
\partial_\delta\prad(\delta,\delta)=\frac{(\lambda+2\mu)}{\delta^{1/3}}>0.
\]
In~\cite{AC18} we have proved that if (and only if) $\delta_c>1$ there exists a unique regular static self-gravitating ball of Seth elastic matter with $\delta(0)=\rho_c/\mathcal{K}=\delta_c$. It follows by Theorem~\ref{regtheo} in the present paper that these balls are strongly regular.

\section{Assumptions on the constitutive functions}\label{assumptions}
Before presenting the assumptions on the constitutive functions required in Theorem~\ref{maintheo}, it is convenient to express the constitutive equations and the system~\eqref{CPSS_Rho} in terms of the new variables 
 \[
 x(r)=\eta(r)^{1/3},\quad y(r)=\frac{\delta(r)}{\eta(r)}.
 \]
 Let
 \begin{equation}\label{newconsteqs}
 p_\mathrm{rad}(r)=\Prad(x(r),y(r)),\quad p_\mathrm{tan}(r)=\Ptan(x(r),y(r))
  \end{equation}
  be the constitutive equations of the elastic ball in the variables $x,y$; we always assume that $\Prad,\Ptan\in C^2((0,\infty)^2)$. For hyperelastic materials~\eqref{hyperDef} gives
  \begin{equation}\label{hyperxy}
  \mathcal{P}_\mathrm{rad}(x,y)= x^3 y^2\,\partial_y \mathcal{W}(x,y),\quad
\mathcal{P}_\mathrm{tan}(x,y)=\frac{1}{2} x^3y (x\partial_x \mathcal{W}(x,y)-y\partial_y\mathcal{W}(x,y)),
  \end{equation}
  where $\mathcal{W}(x,y)=\widehat{w}(\delta,\eta)$ and $\mathcal{W}\in C^3((0,\infty)^2)$.

  {\it Remark.} We choose the power 1/3 in the definition of $x$ so that the constitutive functions $\Prad,\Ptan$ for the elastic materials in Section~\ref{Examples} are rational functions of $(x,y)$. Any other exponent would work.

  In terms of the variables $x,y$, the system~\eqref{CPSS_Rho} reads
  \begin{subequations}\label{newsystem}
\begin{align}
&r\frac{dx}{dr}=-x(1-y),\\
&r\partial_y \Prad(x,y)\frac{dy}{dr}=x\partial_x\Prad(x,y)(1-y)+2(\Ptan(x,y)-\Prad(x,y))-\frac{4\pi\mathcal{K}^2}{3}r^2x^6y.
\end{align}
\end{subequations}

\begin{definition}\label{regdef}
A solution $(x,y)$ of~\eqref{newsystem} in the interval $[0,r_*)$ will be called regular if $(x,y)\in C^0([0,r_*))\cap C^1((0,r_*))$,  $x(r),y(r)>0$, for $r\in [0,r_*)$,
and $\lim_{r\to 0^+}y(r)=1$.
\end{definition}


We are now ready to introduce our assumptions on the constitutive functions. 
It will be shown in Section~\ref{powerlawmaterials} that some of these assumptions are always satisfied by power-law hyperelastic materials,  see  Proposition~\ref{powerprop}.
We begin by requiring some standard conditions in elasticity theory~\cite{AC18,AC19}.
\begin{assumption}\label{As1}
	The constitutive functions~\eqref{newconsteqs} are such that the reference configuration of the elastic ball is stress-free:
	\[
	\Prad(1,1)=\Ptan(1,1)=0,
	\]
	%
	and linear elasticity applies near the reference configuration: 
	%
	\begin{align*}
	&\partial_y\Prad(1,1)=\lambda+2\mu, &\partial_x\Prad(1,1)=3\lambda+2\mu,\\
	&\partial_y\Ptan(1,1)=\lambda, &\partial_x\Ptan(1,1)=3\lambda+2\mu,
	\end{align*}
where the Lam\'e coefficients $\lambda,\mu$ satisfy $\mu\geq0$ and $\lambda+2\mu>0$.  
	%
	\end{assumption}
%
	
{\it Remark.}
	The Lam\'e coefficients in the constitutive functions depend on the material making up the body, i.e., whether it is rubber, steel, copper, etc., see~\cite[p.~129]{Ciarlet83} for a table of values. A body made of a given material (i.e., with given parameters $\lambda,\mu)$ is described by different constitutive equations depending on the amount of strain in the body. Assumption~\ref{As1} requires that for infinitesimal strain (small deformations compared to the size of the body) linear elasticity theory applies, see~\cite{AC18,AC19} for more details.

{\it Remark.} The parameter $\mu$ is also called shear modulus of the material; for fluids, and only in this case, it is given by $\mu=0$. When $\mu>0$ the speed of linear shear (or transversal) elastic waves is defined as $c_\mathrm{T}=\sqrt{\mu/\mathcal{K}}$. The bound $\lambda+2\mu>0$ is also a standard condition on the Lam\'e coefficients (strong ellipticity~\cite{MH}) and implies in particular that the speed of the linear longitudinal elastic waves $c_\mathrm{L}=\sqrt{(\lambda+2\mu)/\mathcal{K}}$ is well-defined.

%
%

%
Next we require that the regular center condition $p_\mathrm{rad}(0)=p_\mathrm{tan}(0)$ in Definition~\ref{balldef} should be satisfied for all possible values of the center density $\rho(0)$, that is to say $\chi(\delta)=\prad(\delta,\delta)-\ptan(\delta,\delta)=\Prad(\delta^{1/3},1)-\Ptan(\delta^{1/3},1)=0$, for all $\delta>0$. 
\begin{assumption}\label{samepcenter}
There holds $\Prad(x,1)=\Ptan(x,1)$, for all $x>0$. 
\end{assumption}
 
It is convenient to define 
\begin{equation}\label{Thetadef}
\Theta(x,y)=\frac{\Ptan(x,y)-\Prad(x,y)}{1-y},
\end{equation}
which, due to Assumptions~\ref{As1}-\ref{samepcenter}, satisfies $\Theta(x,y)\in C^0((0,\infty)^2)$ and
\[ 
\Theta(x,1)=\partial_y\Prad(x,1)-\partial_y\Ptan(x,1),\quad \Theta(1,1)=2\mu.
\] 
In the following assumption we impose that in the limits $x,y\to 0^+$ the material model should behave as a power-law material.
\begin{assumption}\label{gamma-ups}
There exist (necessarily unique) $a\in\R$, $b\in\R$, $c\in\R$ such that the functions
\begin{align*}
&\Gamma(x,y)=(\lambda+2\mu)^{-1}x^ay^{b}\partial_y\Prad(x,y),\\
&\Upsilon(x,y)=(\lambda+2\mu)^{-1}x^a y^{b-1+c}\left(x\partial_x\Prad(x,y)+2\Theta(x,y)\right)
\end{align*}
satisfy $\Gamma,\Upsilon\in C^0([0,\infty)^2)$ and $\Gamma(0,0)\neq0$, $\Upsilon(0,0)\neq0.$
\end{assumption}


We need to impose restrictions on the parameters $a,b,c$ introduced in Assumption~\ref{gamma-ups}.

\begin{assumption}\label{abcass}
The exponents $a,b,c$ in Assumption~\ref{gamma-ups} satisfy
\[
a>-4,\quad b>0,\quad c=0.
\]
\end{assumption}


{\it Remark.} The conditions $b>0$ and $c=0$ imposed in Assumption~\ref{abcass} are the most significant restrictions on the constitutive functions.
As shown in Proposition~\ref{powerprop} in Section~\ref{powerlawmaterials}, for power-law hyperelastic materials the condition $c=0$ follows by $b>0$. 
In Section~\ref{Examples} we show that Assumption~\ref{abcass} is violated by two important examples of elastic material models, namely the Seth model ($b=-1, c=2$) and the Signorini model ($b=-1, c=1$). For the former type of materials the conclusions of Theorem~\ref{maintheo} still hold, see~\cite{AC18}, which shows that Assumption~\ref{abcass} is stronger than necessary. The problem for Signorini materials is still unsolved, although we conjecture that Theorem~\ref{maintheo} also holds in this case.


To motivate our next assumption we remark that, by Assumption~\ref{As1}, $\Gamma(1,1)=1$ and $\Upsilon(1,1)=3$, hence the inequalities $\Gamma(x,1)>0$ and $\Upsilon(x,1)>0$ are satisfied in a open interval around $x=1$ by continuity. In the following two assumptions we impose that these inequalities can only be violated for large values of $x$.

\begin{assumption}[{\bf i}]\label{initial-reg}
There exists (a necessarily unique) $X_\flat\in (1,\infty]$ such that $\Gamma(x,1)>0$ for all $0<x<X_\flat$ and if $X_\flat<\infty$ then $\Gamma(X_\flat,1)=0$. 
\end{assumption}
\setcounter{assumption}{4} 
\begin{assumption}[{\bf ii}]\label{initial-reg2}
There holds $\Upsilon(x,1)>0$, for all $x\in (0,X_\flat)$.
\end{assumption}
The reason to split Assumption~\ref{initial-reg} in two parts is that in the important case of hyperelastic materials Assumption~\ref{initial-reg2}(ii) follows by Assumptions~\ref{samepcenter} and~\ref{initial-reg}(i), as proved in the following simple lemma.

\begin{lemma}\label{hypergamma}
For hyperelastic materials satisfying Assumption~\ref{samepcenter} the following identity holds
\begin{equation}\label{hypcond}
x\partial_x\Prad(x,1)+2\Theta(x,1)=3\partial_y\Prad(x,1),\text{ i.e., $\Upsilon(x,1)=3\Gamma(x,1)$.}
\end{equation}
In particular Assumption~\ref{initial-reg2}(ii) follows by Assumptions~\ref{samepcenter} and~\ref{initial-reg}(i) when the material is hyperelastic.
\end{lemma}
\begin{proof}
Equation~\eqref{hypcond} is equivalent to 
\[
x\partial_x\Prad (x,1)=\partial_y\Prad(x,1)+2\partial_y \Ptan(x,1).
\]
Using~\eqref{hyperxy} we find
\[
x\partial_x\Prad (x,y)-y\partial_y\Prad(x,y)-2y\partial_y \Ptan(x,y)=2(\Prad(x,y)-\Ptan(x,y)),
\]
hence the identity~\eqref{hypcond} follows by Assumption~\ref{samepcenter}.
\end{proof}

{\it Remark.} For hyperelastic materials satisfying Assumption~\ref{samepcenter}, Equation~\eqref{hypcond} is equivalent to $\partial_\eta\widehat{p}_\mathrm{iso}(\delta,\delta)=0$, where $\widehat{p}_\mathrm{iso}(\delta,\eta)$ is the isotropic pressure~\eqref{isopressdef}, hence Lemma~\ref{hypergamma} follows by the last statement in Theorem~\ref{regtheo}. The short proof in the variables $x,y$ is given here to make the proof of Theorem~\ref{maintheo} independent of Theorem~\ref{regtheo}.

In the next assumption we impose additional regularity and some inequalities on the functions $\Gamma,\Upsilon$ defined in Assumption~\ref{gamma-ups}, which will be used to prove global existence and to study the asymptotic behavior of regular solutions to the system~\eqref{newsystem}.

\begin{assumption}\label{hardone}
$\Gamma,\Upsilon\in C^1([0,\infty)^2)$. Moreover there holds
\begin{equation}\label{positivityGamma}
\Gamma(x,y)> 0, \quad\text{for all $(x,y)\in[0,X_\flat)\times[0,1)$}
\end{equation}
and $\Upsilon_0(y)=\Upsilon(0,y)$ satisfies
\begin{equation}\label{dulac-condition}
\Upsilon_0(y)>0,\quad (b\Upsilon_0(y)-y\Upsilon'_0(y))(1-y)+y\Upsilon_0(y)>0,\quad \text{ for all $0<y<1$.}
\end{equation}
\end{assumption}
{\it Remark.} Due to Assumptions~\ref{gamma-ups} and~\ref{initial-reg}(i), the inequalities~\eqref{positivityGamma} and~\eqref{dulac-condition} are also valid respectively for $y=1$ and $y=0$.  


The assumptions thus far are sufficient to prove existence and uniqueness of global regular solutions to the system~\eqref{newsystem}, see Theorem~\ref{global-regular} in Section~\ref{proof}. The elastic balls in Theorem~\ref{maintheo} will be constructed by truncating these regular solutions at a proper finite radius. For this purpose we need two more assumptions.
The first one is a generic requirement which ensures that the tangential pressure is positive in the interior of the ball.

\begin{assumption}\label{ptanpositive}
There exists (a necessarily unique) $X_\sharp \in (1,\infty]$ such that $\Prad(x,y)\geq0$ and $(x,y)\in (0,X_\sharp )\times (0,1)$ imply $\Ptan(x,y)>0$ and if $X_\sharp <\infty$ then $\Prad(X_\sharp ,y_*)=0$ for some $y_*\in (0,1)$ implies $\Ptan(X_\sharp ,y_*)=0$.
\end{assumption}
Let 
\begin{equation}\label{xflat}
X=\min(X_\flat,X_\sharp )\in (1,\infty].
\end{equation}
The elastic balls in Theorem~\ref{maintheo} will be constructed by truncating regular solutions of~\eqref{newsystem} with center datum $x(0)\in (1,X)$. 
The next final assumption is introduced to ensure that the radial pressure can be chosen positive at the center and that it becomes negative for large radii along regular solutions of~\eqref{newsystem}, so that in particular it must vanish at some finite radius if it is positive at the center. Since $x(0)>1$ and, as we prove in Theorem~\ref{global-regular}, $x(r)\to 0$ as $r\to\infty$, it suffices to require the following.

\begin{assumption}\label{neg-press}
 There holds $\Prad(x,1)>0$, for $x\in (1,X)$.
Moreover 
\begin{equation}\label{neg}
\Prad(1,y)<0,\quad\text{for all $0<y<1$}.
\end{equation}
\end{assumption}
{\it Remark.} The inequality~\eqref{neg} is stronger than necessary and may be replaced by the condition that the radial pressure becomes negative for some value of $x\in (0,1]$ and any given $y\in (0,1)$. We choose to express Assumption~\ref{neg-press} in this simple form because~\eqref{neg} is easily provable for all examples of elastic materials in Section~\ref{Examples}. We also remark that~\eqref{neg} implies the bound $m(r)>(4\pi/3)\mathcal{K}r^3$ (i.e., $x(r)>1$) in the interior of the ball which is claimed in~\eqref{bounds}, hence if~\eqref{neg} is weakened this bound might not be true anymore.

\section{Proof of Theorem~\ref{maintheo}}\label{proof}
Theorem~\ref{maintheo} will be proved as a simple corollary of the following result.
\begin{theorem}\label{global-regular}
Let Assumptions~\ref{As1}--\ref{hardone} hold.
For all $x_c\in (0,X_\flat)$ there exists a unique global regular solution $(x(r),y(r))$ of~\eqref{newsystem} satisfying $\lim_{r\to 0^+}x(r)=x_c$. Moreover $x(r)<x_c$, $y(r)<1$, and $\partial_y\Prad(x(r),y(r))>0$ for all $r>0$,
\begin{equation}\label{limit}
\lim_{r\to\infty}(x(r),y(r))=(0,y_\star),\quad y_\star=\frac{a+4}{a+6},
\end{equation}
and $x(r)=O(r^{-2/(6+a)})$ as $r\to\infty$, where $a>-4$ is defined in Assumption~\ref{gamma-ups}.
\end{theorem}
Before proving Theorem~\ref{global-regular} we show how Theorem~\ref{maintheo} follows from it.
\subsection*{Proof of Theorem~\ref{maintheo}}
 The constant $\Delta$ in the theorem is given by
\begin{equation}\label{deltaflat}
\Delta=X^3,\quad X=\min(X_\flat,X_\sharp )\in (1,\infty],
\end{equation}
see~\eqref{xflat}, where $X_\flat,X_\sharp $ are defined in Assumptions~\ref{initial-reg},~\ref{ptanpositive}. Let $\rho_c,\mathcal{K}>0$ be given such that~\eqref{initialdata} hold with $\Delta$ given by~\eqref{deltaflat}. Let $x_c=\delta_c^{1/3}\in (1,X)$ and let $(x(r),y(r))$ be the global regular solution of~\eqref{newsystem} satisfying $\lim_{r\to 0^+}x(r)=x_c$. This solution satisfies $x(r)<x_c<X$ and $y<1$, for all $r>0$, see Theorem~\ref{global-regular}. Define $\delta(r)=x(r)^3y(r)>0$, $\eta(r)=x(r)^3>0$, which is the unique global regular solution of~\eqref{CPSS_Rho} satisfying $\lim_{r\to 0^+}\delta(r)=\lim_{r\to 0^+}\eta(r)=\delta_c$, and let 
\begin{align*}
&\overline{p}_\mathrm{rad}(r)=\widehat{p}_\mathrm{rad}(\delta(r),\eta(r))=\Prad(x(r),y(r)),\\ 
&\overline{p}_\mathrm{tan}(r)=\widehat{p}_\mathrm{tan}(\delta(r),\eta(r))=\Ptan(x(r),y(r)).
\end{align*}
By Assumption~\ref{neg-press} we have $\overline{p}_\mathrm{rad}(0)=\widehat{p}_\mathrm{rad}(\delta_c,\delta_c)=\Prad(x_c,1)>0$. Moreover by~\eqref{limit}, and since $x_c=x(0)>1$, there exists $R_*$ such that $x(R_*)=1$ and thus $\overline{p}(R_*)=\Prad(1,y(R_*))<0$ by~\eqref{neg} in Assumption~\ref{neg-press}. It follows that there exists $R\in(0,R_*)$ such that $\overline{p}_\mathrm{rad}(r)>0$ for $r\in [0,R)$ and $\overline{p}_\mathrm{rad}(R)=0$. Since $(x(r),y(r))\in (0,X)\times(0,1)$, then by Assumption~\ref{ptanpositive} we also have $\overline{p}_\mathrm{tan}(r)>0$, for all $r\in [0,R)$. Hence letting $\overline{\rho}(r)=\mathcal{K}\delta(r)$, we conclude that $(\rho,p_\mathrm{rad},p_\mathrm{tan})=(\overline{\rho},\overline{p}_\mathrm{rad},\overline{p}_\mathrm{tan})\mathbb{I}_{r\leq R}$ is, according to Definition~\ref{balldef}, a regular static self-gravitating elastic ball supported in the interval $r\in [0,R]$. Next we prove the bounds~\eqref{bounds}. The inequality $\partial_\delta\prad(\delta(r),\eta(r))>0$ in the interior of the ball follows from the inequality $\partial_y\Prad(x(r),y(r))>0$, for all $r\geq0$, proved in Theorem~\ref{global-regular}. As $p_\mathrm{rad}(0)=\Prad(x_c,1)>0$, $x_c>1$ and $\Prad(1,y)<0$, for all $0<y<1$ (see Assumption~\ref{neg-press}), then $x(r)>1$ must hold in the interior of the ball, which gives 
\begin{equation}\label{mr}
	m(r)>\frac{4\pi}{3}\mathcal{K}r^3,\quad r\in (0,R].
\end{equation}
Moreover the inequalities $x(r)< x_c$ and $y(r)<1$ proved in Theorem~\ref{global-regular} are equivalent to
\begin{equation}\label{mr2}
\frac{4\pi}{3}\rho(r)r^3< m(r)< \frac{4\pi}{3}\pi\rho_cr^3,
\end{equation} 
which together with~\eqref{mr} proves the second and third inequality in~\eqref{bounds}. \qed

\subsection*{Proof of Theorem~\ref{global-regular}}
One crucial step in the proof is to transform~\eqref{newsystem} into an autonomous dynamical system by replacing $r>0$ with a new independent variable. Before defining this new variable, we remark that within the interval of existence of regular solutions and as long as $\partial_y\Prad(x(r),y(r))$ remains positive we can rewrite~\eqref{newsystem} as
\begin{align*}
&r\Gamma(x,y)\frac{dx}{dr}=-x\Gamma(x,y)(1-y),\\
&r\Gamma(x,y)\frac{dy}{dr}=[\Upsilon(x,y)(1-y)-v]y,
\end{align*}
where $\Gamma,\Upsilon$ are the functions defined in Assumption~\ref{gamma-ups} and, using the condition $c=0$ in Assumption~\ref{abcass},
\begin{equation}\label{v}
v(r)=\frac{4\pi\mathcal{K}^2}{3(\lambda+2\mu)}r^2x(r)^{6+a}y(r)^{b},
\end{equation}
which satisfies
\[
r\Gamma(x,y)\frac{dv}{dr}=[b(\Upsilon(x,y)(1-y)-v)+\Gamma(x,y)(2-(6+a)(1-y)]v.
\]
Recall that $a,b$ are constants satisfying $a>-4$ and $b>0$,
see Assumption~\ref{abcass}.
This suggests to replace the radial variable $r>0$ with the new dimensionless independent variable $\xi\in\R$ defined by
\begin{equation}\label{newr}
\frac{dr}{d\xi}=r\Gamma(x(r),y(r))
\end{equation}
and thus transforms the system~\eqref{newsystem} into the following non-linear autonomous dynamical system:
\begin{subequations}\label{dyn-sys}
\begin{align}
&\frac{dx}{d\xi}=-\Gamma(x,y)(1-y)x,\\
&\frac{dy}{d\xi}=[\Upsilon(x,y)(1-y)-v]y,\\
&\frac{dv}{d\xi}=[b(\Upsilon(x,y)(1-y)-v)+\Gamma(x,y)(2-(6+a)(1-y))]v,
\end{align}
\end{subequations}
where by abuse of notation we use the same symbol to denote functions evaluated in $r>0$ or $\xi\in\R$ (e.g., $x(\xi)=x(r)$).   We emphasize that not all orbits of~\eqref{dyn-sys} correspond to regular solutions of~\eqref{newsystem}.
In particular, the condition $y(r)\to 1$  as $r\to 0^+$ for regular solutions of~\eqref{newsystem} is equivalent to $\lim_{\xi\to-\infty}y(\xi)=1$ for orbits of the dynamical system~\eqref{dyn-sys}.

The state space for the dynamical system~\eqref{dyn-sys} is $(0,\infty)^3$ and by Assumption~\ref{hardone} the flow has a $C^1$ extension on the boundary. The flow can also be extended smoothly on $(x,y,v)\in (-\epsilon,\infty)^3$, for some $\epsilon>0$, by continuing the functions $\Gamma,\Upsilon$ for $x,y<0$. Of course, the dynamics of orbits contained in $[0,\infty)^3$ is not affected by this extension. However the possibility of extending the flow for negative values of $x$, $y$ and $v$ is important to justify the local stability analysis of the fixed points on the boundary. 

Proving local existence and uniqueness of regular solutions for the system~\eqref{newsystem} is equivalent to show that for all $x_c\in (0,X_\flat)$ there is exactly one orbit $\gamma_{x_c}(\xi)=(x(\xi),y(\xi),v(\xi))$ of the dynamical system~\eqref{dyn-sys} such that 
\begin{equation}\label{alphalimit}
\lim_{\xi\to-\infty}\gamma_{x_c}(\xi)=(x_c,1,0).
\end{equation}
To prove this, we begin by studying the local stability properties of the segment of fixed points 
\[
\mathrm{L}_c=(x_c,1,0)\quad x_c\in (0,X_\flat).
\] 
Let $f(x,y,v)$ be the vector field in the right hand side of~\eqref{dyn-sys}. The eigenvalues of $\nabla f(x_c,1,0)$ are
\[
\lambda_1=0,\quad \lambda_2=-\Upsilon(x_c,1),\quad \lambda_3=2\Gamma(x_c,1).
\] 
The corresponding eigenvectors are
\[
e_1=(1,0,0),\quad e_2=(-x_c\Gamma(x_c,1),\Upsilon(x_c,1),0),\quad e_3=(-\textfrac{x_c}{2},-1,2\Gamma(x_c,1)+\Upsilon(x_c,1)).
\]
As $\Gamma(x_c,1)>0$ and $\Upsilon(x_c,1)>0$ (see Assumption~\ref{initial-reg}), the line of fixed points $\mathrm{L}_c$ is a normally hyperbolic manifold of equilibria, see e.g.~\cite[Prop.~4.1]{Aul84}.
The local theory of normally hyperbolic invariant manifolds establishes the local existence and uniqueness of two flow-invariant 2-dimensional manifolds, a stable one $W^s$ and an unstable one $W^u$, consisting of orbits approaching and, respectively, straying from $\mathrm{L}_c$. Moreover the (un)stable manifold $W^u$ is foliated by the (un)stable manifolds of each fixed point on $\mathrm{L}_c$, see~\cite[Prop.~4.1]{HPS} and~\cite{KP90}.
Each unstable manifold in the foliation is tangent to $e_3$ at $(x_c,1,0)$---and thus intersects the interior $(0,\infty)^3$ of the state space---while each stable manifold is tangent to $e_2$ at $(x_c,1,0)$---and so it does not intersect the region $(0,\infty)^3$. It follows that each fixed point $(x_c,1,0)$ for $x_c\in (0,X_\flat)$ is the $\alpha$-limit of a unique orbit $\gamma_{x_c}=(x(\xi),y(\xi),v(\xi))$ in the interior of the state space, as claimed.
This orbit corresponds to the regular solution $(x(r),y(r))$ of~\eqref{newsystem} with center datum $x(0)=x_c$ and up to the maximal radius $r_*=r(\xi_*)$ such that $\Gamma(x(\xi),y(\xi))$ remains positive for $\xi\in [0,\xi_*)$. It will now be shown that $\Gamma(x(\xi),y(\xi))>0$ for all $\xi\in\R$, hence $r_*=\lim_{\xi\to\infty}r(\xi)$.
Consider the regions 
\[
\mathcal{U}=(0,\infty)\times(0,1)\times(0,\infty),\quad \mathcal{U}_{x_c}=(0,x_c)\times(0,1)\times(0,\infty).
\]
Since $(dy/d\xi)_{y=1}=-v<0$ for all $x,v>0$, then  $\mathcal{U}$ is future invariant.
By Assumption~\ref{hardone}, $\Gamma(x,y)>0$, for all $x\in (0,X_\flat)$, $y\in(0,1)$, and $v>0$. Thus $dx/d\xi<0$ holds in the region $\mathcal{U}_{x_c}$ and therefore $\mathcal{U}_{x_c}$ is also future invariant for all $x_c<X_\flat$. It follows that $\gamma_{x_c}\subset\mathcal{U}_{x_c}$ and therefore $\Gamma(x(\xi),y(\xi))$ remains positive along the entire orbit. In particular, $r_*=\lim_{\xi\to\infty}r(\xi)$. Note that this does not yet imply that the regular solution of~\eqref{newsystem} corresponding to $\gamma_{x_c}$ is global. 
Our next goal is to prove that
\begin{equation}\label{limits}
\lim_{\xi\to\infty}x(\xi)=0,\quad \lim_{\xi\to\infty}y(\xi)=y_\star,
\end{equation}
and
\begin{equation}\label{limitv}
\lim_{\xi\to\infty}v(\xi)=v_\star:=\frac{2\Upsilon_0(y_\star)}{a+6},\quad \Upsilon_0(y)=\Upsilon(0,y).
\end{equation}
If~\eqref{limits} hold, then, by~\eqref{newr} and the bound $\Gamma(0,y)>0$ in Assumption~\ref{hardone}, we obtain $dr/d\xi\sim \Gamma(0,y_\star) r>0$ as $\xi\to\infty$, which entails $r_*=\lim_{\xi\to\infty}r(\xi)=\infty$. Hence the regular solution of~\eqref{newsystem} corresponding to $\gamma_{x_c}$ is global and so~\eqref{limits} are equivalent to~\eqref{limit}. Moreover, $v_\star>0$ by the bound $\Upsilon(0,y)>0$ in Assumption~\ref{hardone} and so, by the definition~\eqref{v} of the variable $v$, if~\eqref{limitv} hold then $x(r)\sim r^{-2/(6+a)}$, as $r\to\infty$. Thus the proof of the theorem is complete if we show that~\eqref{limits}-\eqref{limitv} hold.

We begin by observing that $v(\xi)$ is bounded, because
$dv/d\xi\leq (D-bv)v$, where
$D=\sup\{b|\Upsilon(x,y)|+2\Gamma(x,y),\ (x,y)\in(0,x_c)\times(0,1)\}$ (note that we are using $b>0$ here).  
Moreover since $x(\xi)$ is decreasing, then the limit $x_\infty=\lim_{\xi\to\infty}x(\xi)<x_c$ exists. As $|d^2x/d\xi^2|$ is bounded, then $dx/d\xi$ is uniformly continuous and so $\lim_{\xi\to\infty} x'(\xi)=0$ must hold. In particular, 
\begin{equation}\label{temp}
\lim_{\xi\to\infty}\frac{dx}{d\xi}(\xi)=\lim_{\xi\to\infty}x_\infty\Gamma(x_\infty,y(\xi))(y(\xi)-1)=0.
\end{equation}
Assume now that $x_\infty>0$. Since, by Assumptions~\ref{hardone}, $\Gamma(x_\infty,y)>0$ for $y\in [0,1)$, then~\eqref{temp} implies that $\lim_{\xi\to\infty}y(\xi)=1$ must hold. 
By the proven fact that $r(\xi)\to r_*>0$ as $\xi\to\infty$ and that $v(\xi)$ is bounded, we infer that $v(\xi)\to v_\infty\in (0,\infty)$, see~\eqref{v}. (We also infer that $r_*<\infty$, but we shall not make use of this fact.) It follows that $dy/d\xi\to-v_\infty<0$, as $\xi\to\infty$, which contradicts the fact that $y\to 1^-$. 
We conclude that $\lim_{\xi\to\infty}x(\xi)=0$ and thus the $\omega$-limit set of $\gamma_{x_c}$ coincides with the $\omega$-limit of the projection of $\gamma_{x_c}$ onto the boundary region $\{x=0,0< y<1\}$. The flow induced on this region is described by the 2-dimensional dynamical system
\begin{subequations}\label{2dim}
\begin{align}
&\frac{dy}{d\xi}=[\Upsilon_0(y)(1-y)-v]y,\\
&\frac{dv}{d\xi}=[b(\Upsilon_0(y)(1-y)-v)+\Gamma_0(y)(2-(6+a)(1-y))]v,
\end{align}
\end{subequations}
where  $\Gamma_0(y)=\Gamma(0,y)$ and $(y,v)\in (0,1)\times (0,\infty):=\mathcal{V}$. We are only interested in the orbits $\gamma=(y(\xi),v(\xi))$ of~\eqref{2dim} which are entirely contained in $\mathcal{V}$. As $v(\xi)$ is bounded along these orbits, then the $\omega-$limit set $\omega(\gamma)$ is non-empty and by the Poincar\'e-Bendixson theorem $\omega(\gamma)$ is either a periodic orbit, a fixed point, or a connected set consisting of homoclinic and heteroclinic orbits connecting fixed points. We begin by showing that there are no periodic orbits within $\mathcal{V}$. To this purpose we show that
\[
\phi(y,v)=\frac{1}{y^{1+b}v}>0
\]
is a strong Dulac function. In fact, letting $F(y,v)$ be the vector field in the right hand side of~\eqref{2dim}, and using~\eqref{dulac-condition}  in Assumption~\ref{hardone}, we find
\[
\nabla \cdot (\phi F)(y,v)=-\phi(y,v)((b\Upsilon_0(y)-y\Upsilon'_0(y))(1-y)+y\Upsilon_0(y))<0.
\]
Next we study the fixed points of the dynamical system~\eqref{2dim} in the region $\mathcal{V}$. 
In the interior we have the fixed point
\begin{equation}\label{attractor}
\mathrm{P}: (y,v)=(y_\star,v_\star)=\left(\frac{a+4}{a+6},\frac{2\Upsilon_0(y_\star)}{a+6}\right).
\end{equation}
Let again $F(y,v)$ denote the vector field in the right hand side of~\eqref{2dim}. Using the bound $\Gamma_0(y)>0$, $\Upsilon_0(y)>0$ in Assumption~\ref{hardone} together with~\eqref{dulac-condition} we find
\begin{align*}
&\det\nabla F(y_\star,v_\star)=2y_\star\Gamma_0(y_\star)\Upsilon_{0}(y_\star)>0,\\
&\mathrm{Tr}\nabla F(y_\star,v_\star)=-[(b\Upsilon_{0}(y_\star)-y_\star\Upsilon'_0(y_\star))(1-y_\star)+y_\star\Upsilon_0(y_\star)]<0.
\end{align*}
It follows that the fixed point $\mathrm{P}$ is a hyperbolic sink.
On the boundary of $\mathcal{V}$ we find the fixed points
\[
\mathrm{Q}_0: (y,v)=(y_0,v_0)=(0,0),\quad \mathrm{Q}_1: (y,v)=(y_1,v_1)=(1,0)
\]
and if
\begin{equation}\label{sharp}
\Upsilon_0(0)>\frac{a+4}{b}\,\Gamma_0(0)
\end{equation}
we also have the fixed point
\[
\mathrm{Q}_2: (y,v)=(y_2,v_2)=(0,\Upsilon_0(0)-b^{-1}(a+4)\Gamma_0(0)).
\]
The fixed point $\mathrm{Q}_1$ is the end point of the segment $\mathrm{L}_c$ on the invariant boundary $\{x=0\}$. Thus, the local stability properties of $\mathrm{Q}_1$ can be inferred by the analysis of $\mathrm{L}_c$ restricted to $x_c=0$, which gives that 
$\mathrm{Q}_1$ is the $\alpha$-limit of unique orbit in the interior of $\mathcal{V}$. To study the local stability properties of $\mathrm{Q}_0$, $\mathrm{Q}_2$ we first analyze the flow induced on $\{v=0\}$ and $\{y=0\}$.
As $(dy/d\xi)_{v=0}=\Upsilon_0(y)(1-y)y>0$  (by Assumption~\ref{hardone}), the boundary $\{v=0\}$ consists of an orbit connecting $\mathrm{Q}_0$ to $\mathrm{Q}_1$. Since
\[
(dv/d\xi)_{y=0}=(b(\Upsilon_0(0)-v)-(4+a)\Gamma_0(0))v
\]
then, if~\eqref{sharp} is violated, the flow on $\{v=0\}$ consists of an orbit converging to $\mathrm{Q}_0$. Hence in this case $\mathrm{Q}_0$ cannot be the $\alpha$ or $\omega$-limit point of an interior orbit.  
If~\eqref{sharp} holds, i.e., if $\mathrm{Q}_2$ is present, the boundary $\{v=0\}$ consists of two orbits: one connecting $\mathrm{Q}_0$ to $\mathrm{Q}_2$ and one converging to $\mathrm{Q}_2$ from the right (i.e., $\mathrm{Q}_0\rightarrow \mathrm{Q}_2\leftarrow - - $). Hence in this case $\mathrm{Q}_0$ is a source of a 1-parameter family of orbits in the interior of $\mathcal{V}$. Finally, again by Assumption~\ref{hardone}, the eigenvalues of $\nabla F(y_2,v_2)$ are 
$b^{-1}(4+a)\Gamma_0(0)>0$ and $(4+a)\Gamma_0(0)-b\Upsilon(0)<0$, 
and since we have already shown that the stable manifold of $\mathrm{Q}_2$ coincides with $\{v=0\}$, 
then $\mathrm{Q}_2$ must be the source of exactly one orbit in the interior of $\mathcal{V}$. Collecting the information obtained thus far we conclude that the qualitative behavior of the flow of the dynamical system~\eqref{2dim} is as depicted in Figure~\ref{fig:X0}. Due to the absence of periodic orbits, homoclinic  orbits, and heteroclinic contours, the Poincar\'e-Bendixson theorem entails that the $\omega$-limit set of orbits to the dynamical system~\eqref{2dim} in the interior of $\mathcal{V}$ must be the fixed point $\mathrm{P}$, which completes the proof of the theorem. \qed
%

%
\begin{figure}[ht!]
	\centering
	\subfigure[If $\Upsilon_0(0)-\frac{4+a}{b}\Gamma_0(0)<0$]{\label{fig:X0A}
		\includegraphics[width=0.45\textwidth, trim = 0cm 0cm 0cm 0cm]{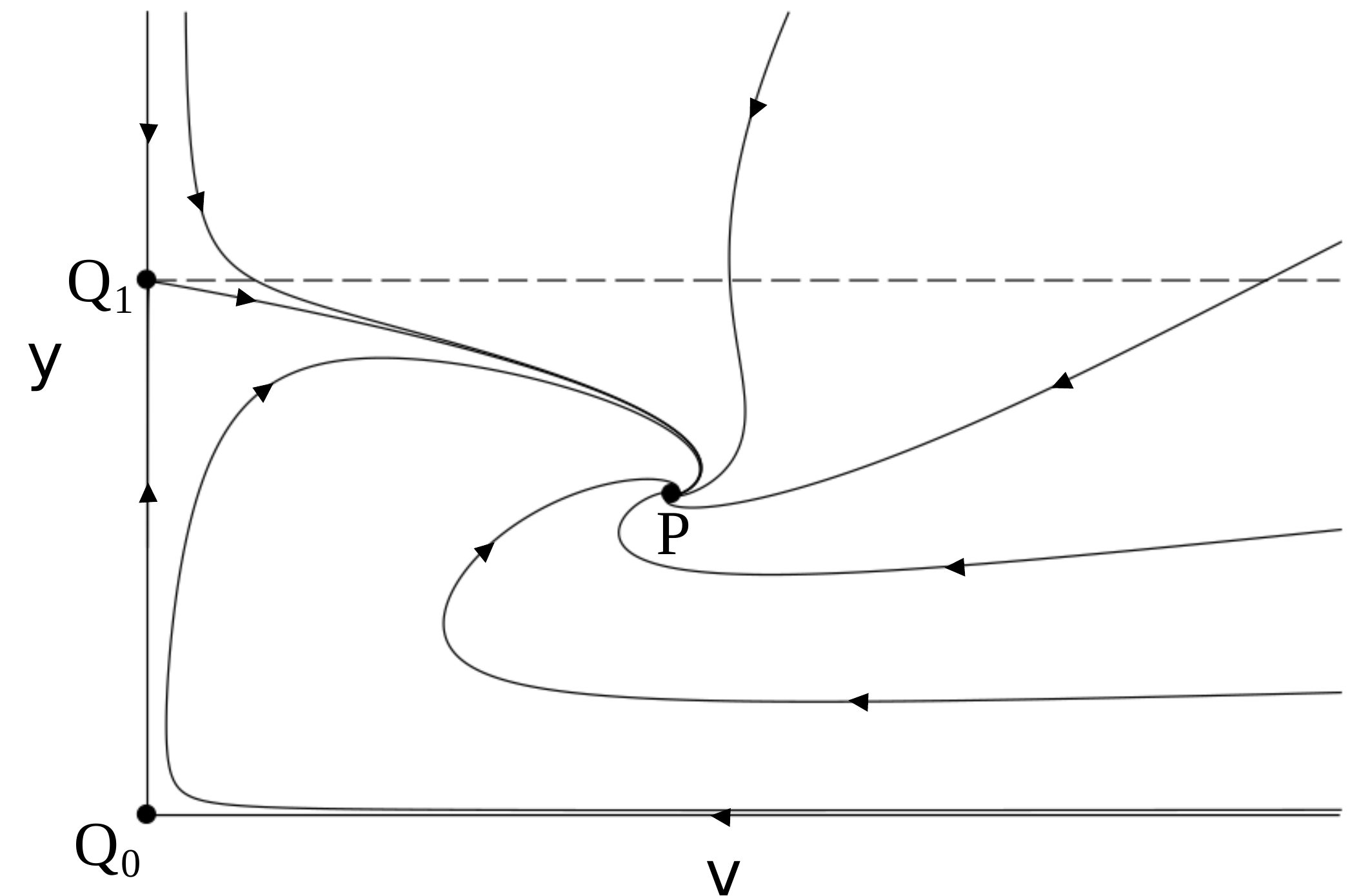}} \hspace{0.1cm}\qquad
	\subfigure[If $\Upsilon_0(0)-\frac{4+a}{b}\Gamma_0(0)>0$ ]{\label{fig:X0B}
		\includegraphics[width=0.45\textwidth, trim = 0cm 0cm 0cm 0cm]{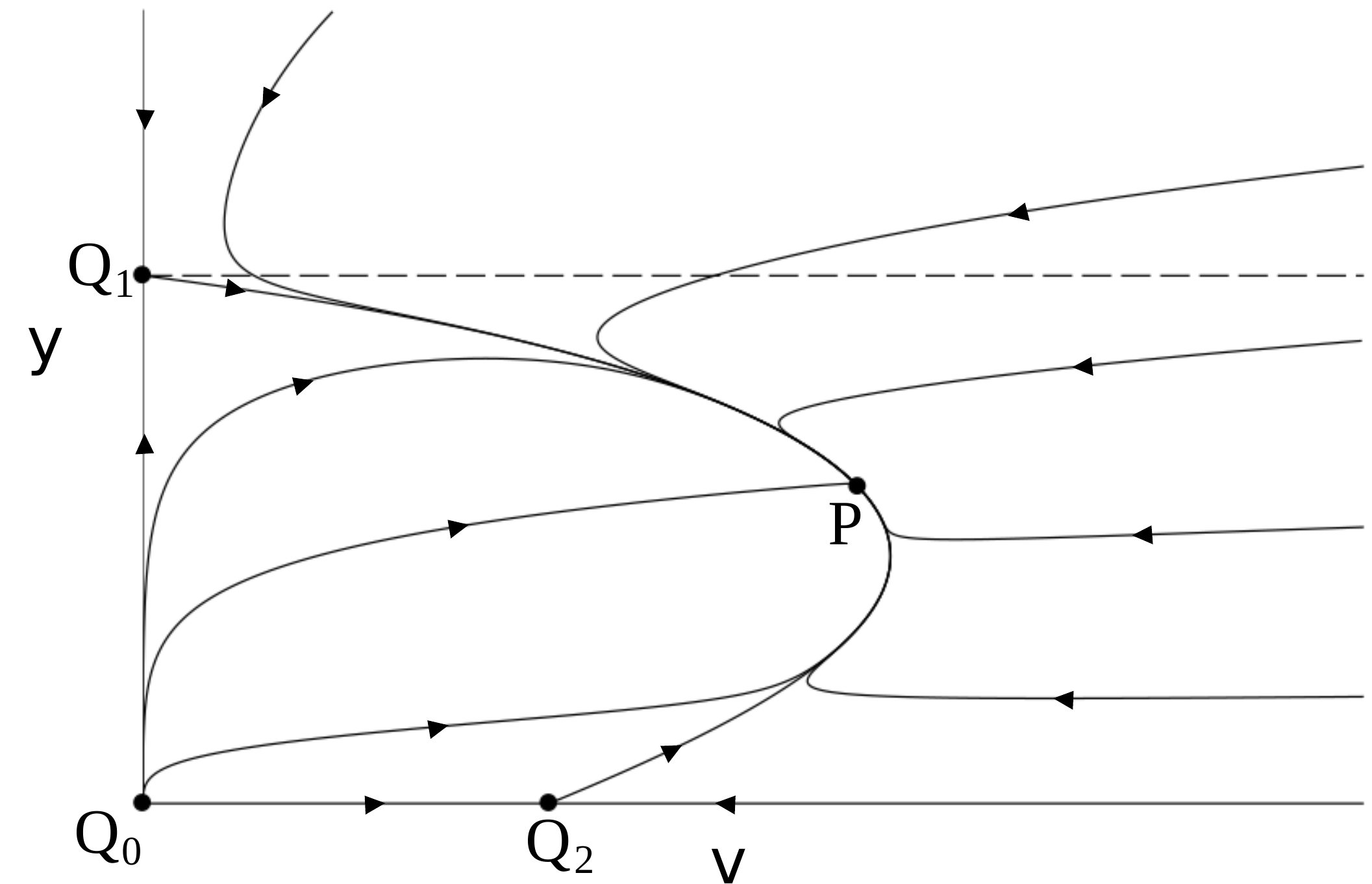}} \hspace{0.1cm}
	\caption{Qualitative behavior of the flow on the invariant boundary $\{x=0\}$.}
	\label{fig:X0}
\end{figure}

\section{Power-law hyperelastic materials}\label{powerlawmaterials}
In this section we study the validity of some of the assumptions in Section~\ref{assumptions} in the case of power-law hyperelastic materials, which are defined as follows.
   \begin{definition}\label{powerdef}
Let $(n_1,n_2,\dots, n_m)\in\mathbb{N}^m$, $m\geq 2$. A stored energy function $\widehat{w}(\delta,\eta)$ is said to be a power-law stored energy function of type $(n_1,\dots,n_m)$ if 
\begin{equation}\label{stored}
\widehat{w}(\delta,\eta)=\mathcal{W}(\eta^{1/3},\delta/\eta),\quad \mathcal{W}(x,y)=\sum_{j=1}^m x^{\gamma_j}\sum_{i=1}^{n_j}\alpha_{ij}y^{\beta_{ij}} +w_0,
\end{equation}
where $w_0,\gamma_j,\beta_{ij},\alpha_{ij}\in\R$ satisfy $\alpha_{ij}\neq 0$, $\gamma_1<\gamma_2<\dots <\gamma_m$, $\beta_{1j}<\beta_{2j}<\dots <\beta_{n_jj}$, for all $j=1,\dots,m$, $i=1,\dots,n_j$, and  if $n_j=1$, then $\gamma_j\neq 0$ and $\beta_{1j}=\gamma_j/3$. Moreover 
\begin{subequations}\label{condsyst}
\begin{align}
&\sum_{j=1}^m\sum_{i=1}^{n_j}\alpha_{ij}=-w_0\label{cond1}\\
&\sum_{j=1}^{m}\sum_{i=1}^{n_j}\alpha_{ij}\gamma_j=0,\quad\sum_{j=1}^{m}\sum_{i=1}^{n_j}\alpha_{ij}\gamma_j^2=3(3\lambda+2\mu),\quad  \sum_{j=1}^m\sum_{i=1}^{n_j}\alpha_{ij}\beta_{ij}^2=\lambda+2\mu,\label{cond2} \\ 
&\sum_{i=1}^{n_j}\alpha_{ij}(\beta_{ij}-\gamma_j/3)=0,\quad j=1,\dots, m\label{cond3}\end{align}
\end{subequations}
and at least one of the exponents $\beta_{ij}$ is different from 0 and $-1$.
\end{definition}
We refer to~\cite{AC19} for a systematic analysis of these stored energy functions. Here we limit ourselves to observe that for a stored energy function of the form~\eqref{stored}, the condition~\eqref{cond1} is equivalent to the normalization condition $\widehat{w}(1,1)=0$ (i.e., the reference configuration is a state of zero energy), while~\eqref{cond2} and~\eqref{cond3} are equivalent to respectively Assumption~\ref{As1} and~\ref{samepcenter} being satisfied. When $n_j=1$,~\eqref{cond3} enforces $\beta_{1j}=\gamma_j/3$ and in this case we require  $\gamma_j\neq 0$ so that $w_0$ is the only zero-degree term in the stored energy function. The last condition on the coefficients $\beta_{ij}$ is required so that $\Prad(x,y)$ is not independent of $y$, i.e., $\prad(\delta,\eta)$ is not independent of $\delta$, see~\eqref{hyperxy} and Equation~\eqref{pradpower} below.

In the next proposition we show that some of our assumptions on the constitutive functions are always satisfied by power-law hyperelastic materials or can be easily verified. To this purpose we use that, by~\eqref{hyperxy}, the constitutive functions of power-law hyperelastic materials are given by
\begin{subequations}\label{Ppower}
\begin{align}
&\mathcal{P}_\mathrm{rad}(x,y)= \sum_{j=1}^{m} x^{3+\gamma_j} \sum_{i=1}^{n_j}\alpha_{ij}\beta_{ij}y^{1+\beta_{ij}},\label{pradpower}\\
&\mathcal{P}_\mathrm{tan}(x,y)=\frac{1}{2}\sum_{j=1}^m x^{3+\gamma_j} \sum_{i=1}^{n_j}\alpha_{ij}(\gamma_j-\beta_{ij})y^{1+\beta_{ij}},
\end{align}
\end{subequations}
by which it follows that
\[
\Theta(x,y)=-\frac{1}{2}\sum_{j=1}^mx^{3+\gamma_j}\sum_{i=1}^{n_j}\alpha_{ij}(3\beta_{ij}-\gamma_j)\frac{y^{1+\beta_{ij}}}{1-y},
\]
see~\eqref{Thetadef}.
Let
\begin{subequations}\label{definitions}
\begin{equation}
I_j:=\{i\in\{1,\dots,n_j\}:\beta_{ij}\neq -1, \beta_{ij}\neq 0\};
\end{equation}
by the last condition in Definition~\ref{powerdef}, $I_j$ is not empty for at least one $j\in\{1,\dots, m\}$. Define
\begin{equation}
\gamma_*=\min\{\gamma_{j}:  j=1,\dots, m, I_j\neq\varnothing\}\,,\quad
\beta_*=\min\{\beta_{ij}: i\in I_j,  j=1,\dots, m, I_j\neq\varnothing\}\,. 
\end{equation}
\end{subequations}
As the exponents $\gamma_1,\gamma_2,\dots $ are increasing, then $\gamma_*=\gamma_p$, where $p$ is the lowest value of $j$ such that $I_j$ is not empty. Denote by $\alpha_{(-1)}$, respectively $\alpha_{(0)}$, the coefficient $\alpha_{ij}$ of the term with exponent $\beta_{ip}=-1$, respectively $\beta_{ip}=0$. If such term does not exist we set  $\alpha_{(-1)}=0$, respectively $\alpha_{(0)}=0$. Moreover we define
\begin{equation}\label{sigma}
\sigma=(3+\gamma_*)\alpha_{(-1)}+\gamma_* \alpha_{(0)}.
\end{equation}


Recall the definitions of $\Gamma(x,y)$ and $\Upsilon(x,y)$ in Assumption~\ref{gamma-ups}. We have the following proposition. 

\begin{proposition}\label{powerprop}
	For power-law hyperelastic materials Assumptions~\ref{As1}-\ref{samepcenter} are satisfied and Assumption~\ref{initial-reg}(ii) follows by Assumption~\ref{initial-reg}(i). Moreover Assumption~\ref{gamma-ups} holds if and only if $\beta_*=\beta_{qp}$ for some (necessarily unique) $q\in I_p$.
	The exponents are given by $a=-3-\gamma_{*}$, 
	$b=-\beta_*$, and 
	\[
	c=
	\left\{\begin{array}{cc} 0 & \text{if $\beta_*<0$ or $\beta_*>0$ and $\sigma = 0$}\\ \beta_* &\text{if $\beta_*>0$ and $\sigma \neq 0$}
	\end{array}\right..
	\]
	In particular, for power-law hyperelastic materials the second condition in Assumption~\ref{abcass}, i.e., $b>0$, implies the third, i.e., $c=0$.
	Finally, the function $\Upsilon_0(y)=\Upsilon(0,y)$ satisfies
	\begin{equation}\label{UAU}
	((b+c)\Upsilon_0(y)-y\Upsilon'_0(y))(1-y)+y\Upsilon_0(y)=y^c((a+3)-ay)\Gamma(0,y),
	\end{equation}
	and thus for $c=0$ and $a\geq -3$ the inequality~\eqref{dulac-condition} follows by the bound $\Gamma(0,y)>0$ in Assumption~\ref{hardone}, while the same inequality is violated for $c=0$ and $a<-3$ when $y$ is sufficiently close to zero. 
\end{proposition}

\begin{proof}
	As already mentioned, Assumption~\ref{As1} and Assumption~\ref{samepcenter} are satisfied by definition of power-law hyperelastic material. The claim on Assumption~\ref{initial-reg} follows by Lemma~\ref{hypergamma}.
	As to the statement concerning Assumption~\ref{gamma-ups}, we first observe that
	\begin{equation}\label{dypradpower}
	\partial_y\Prad(x,y)=\sum_{j=1}^mx^{3+\gamma_j}\sum_{i= 1}^{n_j}\alpha_{ij}\beta_{ij}(1+\beta_{ij})y^{\beta_{ij}}=\sum_{j=1}^mx^{3+\gamma_j}\sum_{i\in I_j}\alpha_{ij}\beta_{ij}(1+\beta_{ij})y^{\beta_{ij}},
	\end{equation}
	where the sum on $i$ is zero when $I_j$ is empty.
	By the definitions of $\gamma_*$ and $\beta_*$,  $3+\gamma_*$ is the lowest exponent of $x$ and $\beta_*$ is the lowest exponent of $y$ in~\eqref{dypradpower}, and so $a=-3-\gamma_*$, $b=-\beta_*$ are the lowest exponents such that $x^ay^{b}\partial_y\Prad(x,y)$ extend continuously for $x=0$ and $y=0$. Moreover if we define $\Gamma$ by using $a>-3-\gamma_*$ and/or $b>-\beta_*$, then the condition $\Gamma(0,0)\neq0$ would be violated. Thus $a=-3-\gamma_*$ and $b=-\beta_*$ must necessarily hold in order that $\Gamma(0,0)\neq 0$. It follows that the function $\Gamma(x,y)$ must be given by
	\begin{equation}\label{semilasttemp}
	\Gamma(x,y)=(\lambda+2\mu)^{-1}\sum_{j=1}^mx^{\gamma_j-\gamma_*}\sum_{i\in I_j}\alpha_{ij}\beta_{ij}(1+\beta_{ij})y^{\beta_{ij}-\beta_*}.
	\end{equation}
	Evaluating~\eqref{semilasttemp} at $x=0$ the only term in the $j$-sum which does not vanish is the $j=p$ term (because $\gamma_*=\gamma_p<\gamma_j$, for $j\neq p$ and $I_j\neq\varnothing$), hence
	\[
	(\lambda+2\mu)\Gamma(0,0)=\sum_{i\in I_p}\alpha_{ip}\beta_{ip}(1+\beta_{ip})(y^{\beta_{ip}-\beta_*})_{y=0}.
	\]
	If there is no $q\in I_p$ such that $\beta_{qp}=\beta_*$, then every exponent of the variable $y$ is positive and thus the above sum is zero. If there exists a (necessarily unique) $q\in\{1,\dots, n_p\}$ such that $\beta_{qp}=\beta_*$, then 
	$(\lambda+2\mu)\Gamma(0,0)=\alpha_{qp}\beta_*(1+\beta_*)\neq 0$.
	This proves the statement about the coefficients $a,b$. 
	The condition on $c$ results from imposing  $\Upsilon(0,0)\neq 0$. Indeed we have
	\begin{align}
	x\partial_x\Prad(x,y)+2\Theta(x,y)&=\sum_{j=1}^m
	x^{3+\gamma_j}(3+\gamma_j)\sum_{i=1}^{n_j}\alpha_{ij}\beta_{ij}y^{1+\beta_{ij}}\nonumber\\
	&\quad-\frac{1}{1-y}\sum_{j=1}^mx^{3+\gamma_j}\sum_{i= 1}^{n_j}\alpha_{ij}(3\beta_{ij}-\gamma_j)y^{1+\beta_{ij}}\nonumber\\
	&=\frac{1}{1-y}\Big[\sum_{j=1}^mx^{3+\gamma_j}\gamma_j\sum_{i=1}^{n_j}\alpha_{ij}(\beta_{ij}+1)y^{1+\beta_{ij}}\nonumber\\
	&\quad-\sum_{j=1}^mx^{3+\gamma_j}(3+\gamma_j)\sum_{i=1}^{n_j}\alpha_{ij}\beta_{ij}y^{2+\beta_{ij}}\Big].\label{lasttemp}
	\end{align}
	Multiplying by $x^a y^{b-1}$, with $a,b$ given as above, we obtain 
	\begin{align}
	[x^ay^{b-1}(x\partial_x\Prad(x,y)+2\Theta(x,y))]&=\frac{1}{1-y}\sum_{j=1}^mx^{\gamma_j-\gamma_*}\Big[\gamma_j\sum_{i=1}^{n_j}\alpha_{ij}(\beta_{ij}+1)y^{\beta_{ij}-\beta_*}\nonumber\\
	&\quad-(3+\gamma_j)\sum_{i=1}^{n_j}\alpha_{ij}\beta_{ij}y^{1+\beta_{ij}-\beta_*}\Big].\label{Upsilon}
	\end{align}
	By~\eqref{cond3}, if $I_j=\varnothing$, then the corresponding $i$-sums in~\eqref{Upsilon} cancel out. In particular all terms for which $\gamma_j-\gamma_*<0$ vanish. Upon evaluating~\eqref{Upsilon} at $x=0$ the only term in the $j$-sum which does not vanish is therefore the $j=p$ term. It follows that
	\begin{align}
	(1-y) [x^ay^{b-1}(x\partial_x\Prad(x,y)+2\Theta(x,y))]_{x=0}&=\gamma_*\sum_{i=1}^{n_p}\alpha_{ip}(\beta_{ip}+1)y^{\beta_{ip}-\beta_*}\nonumber\\
	&\quad-(3+\gamma_*)\sum_{i=1}^{n_p}\alpha_{ip}\beta_{ip}y^{1+\beta_{ip}-\beta_*}\nonumber\\
	&=\sigma y^{-\beta_*}+\gamma_*\sum_{i\in I_p}\alpha_{ip}(\beta_{ip}+1)y^{\beta_{ip}-\beta_*}\nonumber\\
	&\quad-(3+\gamma_*)\sum_{i\in I_p}\alpha_{ip}\beta_{ip}y^{1+\beta_{ip}-\beta_*}.\label{Upsilon2}
	\end{align}
	Evaluating at $y=0$ and using $\beta_*\leq \beta_{ip}$, for all $i\in I_p$, the last term in~\eqref{Upsilon2} vanishes. If $\beta_{*}<0$ or if $\beta_*>0$ and $\sigma=0$, the first term also vanishes at $y=0$. Hence in this case we have $(\lambda+2\mu)\Upsilon(0,0)=\gamma_* \alpha_{qp}(\beta_*+1)\neq 0$.
	When $\beta_*>0$ and $\sigma\neq0$ the first term is singular and thus in order for $\Upsilon$ to be defined at $y=0$ it is necessary to multiply~\eqref{Upsilon} further by the factor $y^c$, $c=\beta_*$, in which case 
	%
	$(\lambda+2\mu)  \Upsilon(0,0)=\sigma \neq0$.
	Now, to prove~\eqref{UAU} we use that the left hand side can be written as
	\begin{equation}\label{proof1}
	((b+c)\Upsilon_0(y)-y\Upsilon'_0(y))(1-y)+y\Upsilon_0(y)=-y^{1+b}\frac{d}{dy}\left(\frac{\Upsilon_0(y)}{y^{b}}(1-y)\right)+c(1-y)\Upsilon_0(y).
	\end{equation}
	By~\eqref{Upsilon2}, 
	\[
	\frac{\Upsilon_0(y)}{y^b}(1-y)=y^c\left[\frac{a}{\lambda+2\mu}\sum_{i\in I_p}^{n_p}\alpha_{ip}\beta_{ip}y^{1+\beta_{ip}}-\frac{3+a}{\lambda+2\mu}\sum_{i\in I_p}\alpha_{ip}(1+\beta_{ip})y^{\beta_{ip}}+\sigma\right],
	\]
	hence
	\begin{equation}\label{proof2}
	y^{1+b}\frac{d}{dy}\left(\frac{\Upsilon_0(y)}{y^{b}}(1-y)\right)=c(1-y)\Upsilon_0(y)+y^{1+b+c}F'(y),
	\end{equation}
	where
	\[
	F(y)=\frac{a}{\lambda+2\mu}\sum_{i\in I_p}^{n_p}\alpha_{ip}\beta_{ip}y^{1+\beta_{ip}}-\frac{3+a}{\lambda+2\mu}\sum_{i\in I_p}\alpha_{ip}(1+\beta_{ip})y^{\beta_{ip}}.
	\]
	Moreover, by~\eqref{semilasttemp},
	\[
	\Gamma(0,y)=(\lambda+2\mu)^{-1}\sum_{i\in I_p}\alpha_{ip}\beta_{ip}(1+\beta_{ip})y^{\beta_{ip}+b},
	\]
	hence $y^{1+b}F'(y)=(a y-(3+a))\Gamma(0,y)$. Thus~\eqref{proof2} becomes
	\[
	y^{1+b}\frac{d}{dy}\left(\frac{\Upsilon_0(y)}{y^{b}}(1-y)\right)=c(1-y)\Upsilon_0(y)+y^{c}(a y-(3+a))\Gamma(0,y).
	\]
	Substituting in~\eqref{proof1} concludes the proof of~\eqref{UAU}.
\end{proof}
{\it Remark.} As power-law hyperelastic materials satisfy Assumption~\ref{samepcenter}, then within this class of material models, any regular static self-gravitating elastic ball is strongly regular, see~Theorem~\ref{regtheo}.

\section{Applications of Theorem~\ref{maintheo} and further results}
\label{Examples}

We now discuss a selection of elastic material models that satisfy the assumptions in Section~\ref{assumptions} and to which therefore Theorem~\ref{maintheo} applies. All these examples belong to the class of power-law hyperelastic materials and thus Propostion~\ref{powerprop} applies as well. Although it is not always necessary, we assume that the Lam\'e coefficients satisfy
	\[
	3\lambda+2\mu>0,
	\]
	i.e., the bulk modulus $\kappa=\lambda+2\mu/3$ is positive, which is a standard condition in elasticity theory~\cite{Ciarlet83,MH} satisfied by all known materials. We refer to~\cite{AC18,AC19} for the derivation of the stored energy functions in this section starting from their original form in Lagrangian  coordinates.

\subsection{Saint Venant-Kirchhoff materials}
Saint Venant-Kirchhoff materials are hyperelastic with type $(3,2)$ power-law  stored energy function given by
\begin{align*}
\mathcal{W}(x,y)&=x^{-4}\left(\textfrac{1}{8}(\lambda+2\mu)y^{-4}+\textfrac{1}{2}\lambda y^{-2}+\textfrac{1}{2}(\lambda+\mu)\right)\\
&\quad+x^{-2}\left(-\textfrac{1}{4}(3\lambda+2\mu)y^{-2}-\textfrac{1}{2}(3\lambda+2\mu)\right)+\textfrac{3}{8}(3\lambda+2\mu).
\end{align*}
%
%
%
%
By Proposition~\ref{powerprop}, Saint Venant-Kirchhoff materials satisfy Assumptions~\ref{As1}-\ref{samepcenter}, as well as Assumption~\ref{gamma-ups} with exponents
\[
a=-3-\gamma_*=1,\quad b=-\beta_*=4,\quad c=0,
\]
where $\gamma_*=\gamma_1=-4$ and $\beta_*=\beta_{11}=-4$, see~\eqref{definitions}.
In particular Assumption~\ref{abcass} is satisfied.
Next we compute
\begin{align*}
\mathcal{P}_\mathrm{rad}(x,y)&= (xy^3)^{-1}(-\textfrac{1}{2}(\lambda+2\mu)-(\lambda -\textfrac{1}{2}(3\lambda+2\mu)x^2)y^{2}), \\
\mathcal{P}_\mathrm{tan}(x,y)&= (xy^3)^{-1}(-\textfrac{1}{2}\lambda y^{2}-((\lambda+\mu) - \textfrac{1}{2}(3\lambda+2\mu)x^2)y^4),\\
\partial_x\mathcal{P}_\mathrm{rad}(x,y)&=(x^2y^3)^{-1}(\textfrac{1}{2}(\lambda+2\mu)+(\lambda+\textfrac{1}{2}(3\lambda+2\mu)x^2)y^2),\\
\partial_y\mathcal{P}_\mathrm{rad}(x,y)&=(xy^4)^{-1}(\textfrac{3}{2}(\lambda+2\mu)+(\lambda-\textfrac{1}{2}(3\lambda+2\mu)x^2 ) y^2).
\end{align*}
As $\partial_y\mathcal{P}_\mathrm{rad}(x,1)=x^{-1}(\textfrac{1}{2}(5\lambda+6\mu)-\textfrac{1}{2}(3\lambda+2\mu)x^2)$, and having assumed $3\lambda+2\mu>0$,
Assumption~\ref{initial-reg} is satisfied with 
\[
X_\flat=\sqrt{\frac{5\lambda+6\mu}{3\lambda+2\mu}}.
\]
Note also that $\partial_y\Prad(X_\flat,1)=0$ and $\partial_y\Prad(x,1)<0$ for $x>X_\flat$, that is $\partial_\delta\prad(X_\flat^3,X_\flat^3)=0$ and $\partial_\delta\prad(\delta_c,\delta_c)<0$ for $\delta_c>X_\flat^3$, hence the condition $\delta_c<X_\flat^3$ is necessary for the strict hyperbolicity condition~\eqref{stricthypcond} to be satisfied at the center. 
The functions $\Gamma,\Upsilon$ are given by
\begin{align*}
\Gamma(x,y)
&=\frac{3}{2}(1-y^2)+\frac{3\lambda+2\mu}{2(\lambda+2\mu)}(X_\flat^2-x^2)y^2, \\
\Upsilon(x,y)&=
\frac{3}{2}(1+\frac{y}{3}(5-y))(1-y)+\frac{3}{2}\frac{3\lambda+2\mu}{\lambda+2\mu}(X_\flat^2-x^2)y^2,
\end{align*}
which yields
\begin{align*}
&\Gamma(x,0)=\frac{3}{2},\quad\quad         \Gamma(0,y)=\frac{3}{2}+\frac{\lambda}{\lambda+2\mu}y^2>0,\quad y\in [0,1), \\
&\Upsilon(x,0)=\frac{3}{2},\quad\quad\Upsilon(0,y)=\frac{3}{2}+y+\frac{3}{2}\frac{3\lambda+2\mu}{\lambda+2\mu} y^2+\frac{y^3}{2}:=\Upsilon_0(y)>0,\quad y\in (0,1).
\end{align*}
%
Since $\Gamma(x,y)>0$ for $(x,y)\in [0,X_\flat)\times[0,1)$ and, by~\eqref{UAU},
\[
(b\Upsilon_0(y)-y\Upsilon_0'(y))(1-y)+y\Upsilon_0(y)=(4-y)\Gamma(0,y),
\]
then Assumption~\ref{hardone} is satisfied as well. At this point we have shown that the hypotheses of Theorem~\ref{global-regular} hold for Saint Venant-Kirchhoff materials. We now prove that the remaining assumptions in Theorem~\ref{maintheo} are also satisfied. Since
\[
\mathcal{P}_\mathrm{tan}(x,y)=y^2\mathcal{P}_\mathrm{rad}(x,y)+\mu(xy)^{-1}(1-y^2),
\]
then Assumption~\ref{ptanpositive} is satisfied with $X_\sharp =\infty$. Let $X=\min(X_\sharp ,X_\flat)=X_\flat$; 
since 
\begin{equation}\label{pcentersvk}
\mathcal{P}_\mathrm{rad}(x,1)=\textfrac{1}{2}(3\lambda+2\mu)x^{-1}(x^2-1),
\end{equation}
and since $\Prad(1,y)=\textfrac{1}{2}(\lambda+2\mu)(y^2-1)<0$, for all $0<y<1$, then Assumption~\ref{neg-press} is satisfied. 
Note that in this case the constant $\Delta$ in Theorem~\ref{maintheo} is finite and given by $\Delta=\Delta_\flat:=X_\flat^{3}$, see~\eqref{deltaflat}. Moreover, by~\eqref{pcentersvk}, when $x_c=\delta_c^{1/3}\leq 1$ (i.e., $\rho_c\leq\mathcal{K}$) the central pressure is $p_\mathrm{rad}(0)=\Prad(x_c,1)\leq0$ and therefore the elastic body is not a ball, which means that the condition $\delta_c>1$ in Theorem~\ref{maintheo} is necessary in the Saint Venant-Kirchhoff model. 
Thus Theorem~\ref{maintheo} for Saint Venant-Kirchhoff materials can be sharpened to the following.
\begin{theorem}
	When the elastic material is given by the Saint Venant-Kirchhoff model, the condition $\delta_c:=\rho_c/ \mathcal{K}>1$ is necessary for the existence of regular static self-gravitating balls, while the condition $\delta_c<\Delta_\flat$, where
	\[
	\Delta_\flat=\left(\frac{5\lambda+6\mu}{3\lambda+2\mu}\right)^{3/2},
	\] 
	 is necessary for the strict hyperbolicity condition~\eqref{stricthypcond} to be satisfied at the center. When $1<\delta_c<\Delta_\flat$ 
	there exists a unique strongly regular static self-gravitating ball with central density $\rho(0)=\rho_c$. Moreover $\partial_\delta\prad(\delta_c,\delta_c)> 0$ and
	\[
\partial_\delta\prad(\delta(r),\eta(r))>0,\quad \rho(r)<\rho_c,\quad \frac{4\pi}{3}\max(\rho(r),\mathcal{K})r^3<m(r)<\frac{4\pi}{3}\rho_c r^3
\]
hold for all $r\in (0,R]$, where $R>0$ is the radius of the ball. 
\end{theorem}
	

\subsection{John (harmonic) materials}
John materials have the following type (1,3,2) power-law stored energy function:
\begin{align*}
\mathcal{W}(x,y)&=-2\mu x^{-3}y^{-1}+x^{-2}\left(\textfrac{1}{2}(\lambda+2\mu)y^{-2}+2(\lambda+2\mu)y^{-1}+2(\lambda+2\mu)\right)\\
&\quad-x^{-1}\left((3\lambda+4\mu)y^{-1}+2(3\lambda+4\mu)\right)+\textfrac{1}{2}(9\lambda+10\mu).
\end{align*}
By Proposition~\ref{powerprop}, John materials satisfy Assumptions~\ref{As1},~\ref{samepcenter} as well as Assumption~\ref{gamma-ups} with exponents
\[
a=-3-\gamma_*=-1,\quad b=-\beta_*=2,\quad c=0,
\]
where $\gamma_*=\gamma_2=-2$ and $\beta_*=\beta_{21}=-2$, see~\eqref{definitions}.
In particular Assumption~\ref{abcass} is satisfied.
Moreover
\begin{align*}
\mathcal{P}_\mathrm{rad}(x,y)&=2\mu+xy^{-1}(-(\lambda +2 \mu )-(2(\lambda +2 \mu )-(3 \lambda +4 \mu )x)y),\\
\mathcal{P}_\mathrm{tan}(x,y)&=2\mu+xy^{-1}(-(\lambda +2 \mu )y-(2(\lambda +2 \mu )-(3 \lambda +4 \mu )x)y^2),\\
\partial_x\mathcal{P}_\mathrm{rad}(x,y)&=y^{-1}(-(\lambda +2 \mu )-(2(\lambda +2 \mu )-2(3\lambda+4\mu)x)y), \\
\partial_y\mathcal{P}_\mathrm{rad}(x,y)&=xy^{-2}(\lambda+2\mu),
\end{align*}
from which we see that Assumption~\ref{initial-reg} is satisfied with $X_\flat=\infty$.
In this case $\Gamma$ and $\Upsilon$ are independent of $x$ and given by
\[
\Gamma(x,y)=1,\quad
\Upsilon(x,y)=1+2y.
\]
By~\eqref{UAU}
\[
(b\Upsilon_0(y)-y\Upsilon_0'(y))(1-y)+y\Upsilon_0(y)=2+y
\]
and therefore Assumption~\ref{hardone} is satisfied. Thus Theorem~\ref{global-regular} holds for John materials. Next we note that $\Ptan(x,y)$ can be rewritten as
%
\begin{equation*}
\mathcal{P}_\mathrm{tan}(x,y)=y\Prad(x,y)+2\mu(1-y),
\end{equation*}
hence Assumption~\ref{ptanpositive} is satisfied with $X_\sharp =\infty$. In particular $X=\min(X_\flat,X_\sharp )=\infty$ and so $\Delta$ in Theorem~\ref{maintheo} is given by $\Delta=\infty$, see~\eqref{deltaflat}. Since
\begin{equation}\label{pcenterjohn}
\mathcal{P}_\mathrm{rad}(x,1)=(2\mu-(3\lambda+4\mu)x)(1-x)>0,\quad \text{for $x\in (0,X_*)\cup(1,\infty)$},
\end{equation}
where
\[
X_*=\frac{2\mu}{3\lambda+4\mu}<1,
\]
and
\[
\Prad(1,y)=-(\lambda+2\mu)y^{-1}(1-y)<0,\quad \text{for all $y\in (0,1)$,}
\]
then Assumption~\ref{neg-press} is satisfied. Thus Theorem~\ref{maintheo} for John materials can be sharpened to the following.
\begin{theorem}\label{johntheo}
	When the elastic material is given by the John model, for all $\delta_c:=\rho_c/\mathcal{K}>1$ there exists a unique strongly regular static self-gravitating ball with central density $\rho(0)=\rho_c$. Moreover $\partial_\delta\prad(\delta_c,\delta_c)> 0$ and
	\[
\partial_\delta\prad(\delta(r),\eta(r))>0,\quad \rho(r)<\rho_c,\quad \frac{4\pi}{3}\max(\rho(r),\mathcal{K})r^3<m(r)<\frac{4\pi}{3}\rho_c r^3
\]
hold for all $r\in (0,R]$, where $R>0$ is the radius of the ball. 
\end{theorem}
	%
{\it Remark.} Due to~\eqref{pcenterjohn}, the radial pressure at the center is also positive if we choose the center datum $x_c$ such that $0<x_c<X_*$, i.e., $\delta_c<\Delta_*$, where $0<\Delta_*:=X_*^3<1$. We have not been able to prove or disprove the existence of finite radius ball solutions for these data, hence the question of whether the condition $\delta_c>1$ in Theorem~\ref{johntheo} is necessary remains open.

{\it Remark.} The John model presented here is a special case of a more general class of hyperelastic power-law harmonic materials, see~\cite{AC19}.

\subsection{Hadamard materials}
The stored energy function $\widehat{w}(\delta,\eta)$ of Hadamard materials is defined up to an additive term  $h(\delta^{-2})$, where $h$ is a positive function. Choosing $h(z)\sim z^k$, with $k$ independent of $\lambda,\mu$, gives the following power-law stored energy function~\cite{AC18,AC19}:
\begin{align*}
\mathcal{W}(x,y)&=x^{-4}\left(\frac{\lambda+2k\mu}{2(1-k)}y^{-2}+\frac{\lambda+2k\mu}{4(1-k)}\right)+x^{-2}\left(\frac{\lambda+2(2k-1)\mu}{4(k-1)}y^{-2}+\frac{\lambda+2(2k-1)\mu}{2(k-1)}\right)\\
&\quad+\frac{\lambda+2\mu}{4k(k-1)}x^{-6k}y^{-2k}+\frac{\lambda+2(1-3k+3k^2)\mu}{4k(1-k)},\quad k\neq0,\quad k\neq 1.
\end{align*}
By Proposition~\ref{powerprop}, Hadamard materials satisfy Assumptions~\ref{As1},~\ref{samepcenter} as well as Assumption~\ref{gamma-ups} with exponents given by
\begin{align*}
&k>1\Rightarrow a=-3+6k,\quad b=2k,\quad c=0,\\
&\textfrac{2}{3}<k< 1\Rightarrow a=-3+6k,\quad b=2, \quad c=0,\\
&k\leq\textfrac{2}{3},\ k\neq 0\Rightarrow a=1,\quad b=2,\quad c=0.
\end{align*}
In particular Assumption~\ref{abcass} is satisfied for all $k\neq 0$, $k\neq 1$.
For the verification of the remaining assumptions we restrict ourselves to the case
\[
k=1/2,
\]
which simplifies considerably the calculations (in particular, the powers in the stored energy function become integers). With this choice, the stored energy function of Hadamard materials takes the form
\begin{align}
\mathcal{W}(x,y)&=x^{-4}\left((\lambda+\mu)y^{-2}+\textfrac{1}{2}(\lambda+\mu)\right)-(\lambda+2\mu)x^{-3}y^{-1}\nonumber\\
&\quad+x^{-2}\left(-\textfrac{1}{2}\lambda y^{-2}-\lambda\right)+\textfrac{1}{2}(2\lambda+\mu),\label{storedhad}
\end{align}
which is of type (2,1,2) with $a=1$, $b=2$ and $c=0$.
We also require
\[
\lambda>0,
\]
which is a condition satisfied by most materials.
For the stored energy function~\eqref{storedhad} we have
\begin{align*}
\mathcal{P}_\mathrm{rad}(x,y)&=(xy)^{-1}(\lambda x^2+(\lambda+2\mu)xy-2(\lambda+\mu)),\\
\mathcal{P}_\mathrm{tan}(x,y)&=(xy)^{-1}(\lambda x^2y^2+(\lambda+2\mu)xy-(\lambda+\mu)y^2-\lambda-\mu),\\
 \partial_x\mathcal{P}_\mathrm{rad}(x,y)&=(x^2y)^{-1}(\lambda x^2+2(\lambda+\mu)),\\
\partial_y\mathcal{P}_\mathrm{rad}(x,y)&=(xy^2)^{-1}(2 \lambda +2 \mu -\lambda  x^2).
\end{align*}
%
%
Assumption~\ref{initial-reg} is satisfied with
\[
X_\flat=\sqrt{2\frac{\lambda+\mu}{\lambda}}.
\]
Moreover $\partial_y\Prad(X_\flat,1)=0$ and $\partial_y\Prad(x,1)<0$ for $x>X_\flat$, that is $\partial_\delta\prad(X_\flat^3,X_\flat^3)=0$ and $\partial_\delta\prad(\delta_c,\delta_c)<0$ for $\delta_c>X_\flat^3$, hence the condition $\delta_c<X_\flat^3$ is necessary for the strict hyperbolicity condition~\eqref{stricthypcond} to be satisfied at the center. 
The functions $\Gamma,\Upsilon$ are
\[
 \Gamma(x,y)=\frac{\lambda}{\lambda+2\mu}(X_\flat^2-x^2),\quad
\Upsilon(x,y)=(2+y)\frac{\lambda}{\lambda+2\mu}(X_\flat^2-x^2)+\frac{\lambda}{\lambda+2\mu}(1-y)x^2.
\]
By~\eqref{UAU}, $\Upsilon_0(y)=\Upsilon(0,y)$ satisfies
\[
\Upsilon_0(y)=2\frac{\lambda+\mu}{\lambda+2\mu}(2+y)>0,\quad (b\Upsilon_0(y)-y\Upsilon_0'(y))(1-y)+y\Upsilon_0(y)=(4-y)\Gamma(0,y),
\]
hence Assumption~\ref{hardone} holds. Thus Theorem~\ref{global-regular} holds true for this particular Hadamard material model. Moreover since
\[
\Ptan(x,y)=(xy)^{-1}(1-y^2)(\lambda+\mu-\lambda x^2)+\Prad(x,y)
\]
then Assumption~\ref{ptanpositive} is satisfied with
\[
X_\sharp =\sqrt{\frac{\lambda+\mu}{\lambda}}.
\]
In particular
\[
X=\min(X_\flat,X_\sharp )=X_\sharp ,
\]
hence $\Delta=\Delta_\sharp:=X_\sharp ^3$ in Theorem~\ref{maintheo}. Finally
\[
\Prad(x,1)=x^{-1}(x-1)(\lambda x+2(\lambda+\mu)),
\]
and $\Prad(1,y)=(\lambda+2\mu)y^{-1}(y-1)<0$, for $y\in (0,1)$, 
which shows that Assumption~\ref{neg-press} holds and that $x_c>1$ (i.e., $\delta_c>1$) is necessary for the existence of elastic balls. Thus we have proved the following final theorem.

\begin{theorem}\label{hadtheo}
When the elastic material is given by the Hadamard model with stored energy function~\eqref{storedhad} and $\lambda>0$, the condition $\delta_c:=\rho_c/ \mathcal{K}>1$ is necessary for the existence of regular static self-gravitating balls. For 
\[
1<\delta_c<\left(\frac{\lambda+\mu}{\lambda}\right)^{3/2}=\Delta_\sharp
\] 
there exists a unique strongly regular static self-gravitating ball with central density $\rho(0)=\rho_c$. Moreover $\partial_\delta\prad(\delta_c,\delta_c)> 0$ and
\[
\partial_\delta\prad(\delta(r),\eta(r))>0,\quad \rho(r)<\rho_c,\quad \frac{4\pi}{3}\max(\rho(r),\mathcal{K})r^3<m(r)<\frac{4\pi}{3}\rho_c r^3
\]
hold for all $r\in (0,R]$, where $R>0$ is the radius of the ball. 
\end{theorem}
{\it Remark.} The bound $\delta_c<\Delta_\sharp$ in Theorem~\ref{hadtheo} might be stronger than necessary. In particular, let 
\[
\Delta_\flat:=X_\flat^3=\left(2\frac{\lambda+\mu}{\lambda}\right)^{3/2}>\Delta_\sharp.
\]
For $\delta_c\geq\Delta_\flat$ the strict hyperbolicity condition~\eqref{stricthypcond} is violated at the center. An interesting open question is whether the bound $\delta_c<\Delta_\sharp$ in Theorem~\ref{hadtheo} can be replaced by $\delta_c<\Delta_\flat$, or more generally if Assumption~\ref{ptanpositive} may be relaxed so that the only restriction on the center datum in Theorem~\ref{maintheo} is provided by Assumption~\ref{initial-reg}.
\subsection{Important examples not covered by Theorem~\ref{maintheo}}
The are of course elastic material models which violate some of our assumptions on the constitutive equations. In this section we discuss briefly two important models which are not covered by Theorem~\ref{maintheo}. We limit ourselves to discuss the violation of Assumption~\ref{abcass}, as this is the most important restriction imposed on the constitutive functions.

The first example is the Seth model, for which the constitutive functions are given by
\[
\Prad(x,y)=\lambda x^2+\frac{\lambda+2\mu}{2}x^2y^2-p_0,\quad \Ptan(x,y)=(\lambda+\mu) x^2+\frac{\lambda}{2}x^2y^2-p_0
\]
where $p_0=(3\lambda+2\mu)/2>0$, see~\eqref{pseth}.
For this model Assumption~\ref{gamma-ups} holds with $a=-2$, $b=-1$ and $c=2$, hence Assumption~\ref{abcass} is violated. 
However the conclusions of Theorem~\ref{maintheo} still hold true for the Seth model, as we proved in~\cite{AC18}. The main difference between the approach in this paper and the approach in~\cite{AC18} is that in~\cite{AC18} we used a different set of dynamical variables for which $y=0$ is no longer an invariant set. It appears that the formulation employed in~\cite{AC18} is more convenient for models with $c>0$.

The second important example not covered by the results of this paper is the Signorini model. In one special case (called quasi-linear), the stored energy function of Signorini materials is given by
\begin{align*}
\mathcal{W}(x,y)&=\frac{9\lambda+5\mu}{8}x^{-3}y^{-1}+x^{-1}\left(-\frac{3\lambda+\mu}{2}y^{-1}-\frac{3\lambda+\mu}{4}y\right)\\
&\quad+x\left(\frac{\lambda+\mu}{2}y^{-1}+\frac{\lambda+\mu}{2}y+\frac{\lambda+\mu}{8}y^3\right)-\mu,
\end{align*}
which is a type (1,2,3) power-law stored energy function with  $a=-2$, $b=-2$ and $c=1$, thus Assumption~\ref{abcass} is violated by the Signorini model. In this case the question of existence of regular static self-gravitating balls remains open.

\subsection{Exact solutions}\label{exactsolutions}
We conclude the paper by discussing the existence of (typically non-regular) exact self-similar solutions for the system~\eqref{dyn-sys}. Exact solutions can be constructed when the functions $\Gamma$ and $\Upsilon$ are independent of $x$, in which case the differential equation for $x$ decouples from the rest of the system, thus generating a skew product flow that further implies the existence of an interior straight orbit $--\rightarrow \mathrm{P}$ for the dynamical system~\eqref{dyn-sys}, where $\mathrm{P}$ is the fixed point~\eqref{attractor} on the boundary $x=0$.
We present below two examples of exact non-regular balls; see~\cite{AC19} for a general discussion on this problem.

\subsubsection*{An exact non-regular static self-gravitating ball in the Seth model}
For the Seth model
the exact self-similar solution has been found in~\cite{AC18}, namely
\[
\delta_\star(r)=\frac{d^{3/4}}{2r^{3/2}},\quad \eta_\star(r)=\frac{d^{3/4}}{r^{3/2}},\quad d=\frac{3 (9\lambda +14 \mu )}{16\pi \mathcal{K}^2}.
\]
The associated principal pressures~\eqref{pseth} are given by
\begin{align*}
&p_\mathrm{rad}^\star(r)=-p_0+\frac{9\lambda+2\mu}{8}\frac{d^{1/2}}{r},\\
& p_\mathrm{tan}^\star(r)=-p_0+\frac{9\lambda+8\mu}{8}\frac{d^{1/2}}{r},
\end{align*}
where $p_0=(3\lambda+2\mu)/2>0$.
Note that $p_\mathrm{tan}(r)>p_\mathrm{rad}(r)$, for all $r>0$, the radial pressure is positive for $r<R$, negative for $r>R$ and vanishes at $r=R$, where
\[
R=\frac{(9 \lambda +2 \mu ) \sqrt{27 \lambda +42 \mu }}{\sqrt{\pi } \mathcal{K}  (48 \lambda +32 \mu )}.
\]
Hence the self-similar solution truncated at $r=R$ describes a static, self-gravitating ball with irregular center.

%
%
\subsubsection*{An exact non-regular static self-gravitating ball in the John model}
Another example of material for which $\Gamma,\Upsilon$ are independent of $x$ is the John model. In this case we find the solution
\[
\delta_\star(r)=\frac{3}{5}\frac{d^{3/5}}{r^{6/5}},\quad \eta_\star(r)=\frac{d^{3/5}}{r^{6/5}},\quad d=\frac{11(\lambda+2\mu)}{6\pi\mathcal{K}^2}.
\]
The associated principal pressures are
	\begin{align*}
	&p_{\mathrm{rad}}^\star(r) = 2\mu-\textfrac{11}{3}(\lambda+2\mu)\left(\frac{d}{r^{2}}\right)^{\frac{1}{5}}+(3\lambda+4\mu)\left(\frac{d}{r^{2}}\right)^{\frac{2}{5}}, \\
	&p_{\mathrm{tan}}^\star(r) = 2\mu-\textfrac{11}{5}(\lambda+2\mu)\left(\frac{d}{r^{2}}\right)^{\frac{1}{5}}+\textfrac{3}{5}(3\lambda+4\mu)\left(\frac{d}{r^{2}}\right)^{\frac{2}{5}}.
	\end{align*}
%
The radial pressure vanishes at the radii
\[
R_{\pm} = d^{\frac{1}{2}} \left(\frac{6(3\lambda+4\mu)}{11(\lambda+2\mu)\pm\sqrt{(11(\lambda+2\mu))^2-72\mu(3\lambda+4\mu)}}\right)^{\frac{5}{2}}.
\]
%
Note that $(11(\lambda+2\mu))^2-72\mu(3\lambda+4\mu)=121\lambda^2+268\lambda\mu+196\mu^2$, hence assuming for instance that $\lambda>0$, which is satisfied by most materials, $R_+<R_-$ are both well defined and we have
\[
p_\mathrm{rad}(R_\pm)=0,\quad p_\mathrm{rad}(r)>0,\quad\text{for $r\in (0,R_+)\cup (R_-,\infty):=J$.}
\]
It is straightforward to verify that $p_\mathrm{tan}(r)>0$, for $r\in J$, thus in the case of the John model the self-similar solution gives rise to a non-regular static self-gravitating ball with radius $R_+$. 

\section*{Acknowledgements}
S.~C. thanks Stephen Pankavich for discussions on the proof of Theorem~\ref{regtheo}.
A.~A. is supported by the project (GPSEinstein) PTDC/MAT-ANA/1275/2014,
by CAMGSD, Instituto Superior T{\'e}cnico, through FCT/Portugal 
UID/MAT/04459/2013, and by the FCT Grant No. SFRH/BPD/85194/2012.


%

\end{document}